\documentclass[runningheads]{llncs} %
\usepackage{graphicx} 

\usepackage[utf8]{inputenc}
\usepackage{amssymb}
\usepackage{threeparttable}
\usepackage{graphicx}
\usepackage{color}
\usepackage{paralist}
\usepackage{framed}
\usepackage{caption}
\usepackage{todonotes}
\usepackage{subfig}
\usepackage{multirow}
\usepackage{amsmath}
\usepackage{hyperref}

\usepackage[shortlabels]{enumitem}

\usepackage{algorithmic}
\usepackage[ruled,vlined]{algorithm2e}
\usepackage{amsfonts}

\SetKw{continue}{continue}

\usepackage{tikz}
\usetikzlibrary{fit,calc,positioning,decorations.pathreplacing,matrix}

\definecolor{pdfurlcolor}{rgb}{0,0,0.6}
\definecolor{pdfcitecolor}{rgb}{0,0.6,0}
\definecolor{pdflinkcolor}{rgb}{0.6,0,0}
\usepackage{epsfig,color}

\makeatletter
\newcommand{\pushright}[1]{\ifmeasuring@#1\else\omit\hfill$\displaystyle#1$\fi\ignorespaces}
\newcommand{\pushleft}[1]{\ifmeasuring@#1\else\omit$\displaystyle#1$\hfill\fi\ignorespaces}
\makeatother

\newcommand{\scomm}{\ensuremath{{\cal S}_{\text{COMM}}}}
\newcommand{\scomp}{\ensuremath{{\cal S}_{\text{COMP}}}}

\newcommand{\threepart}{\textsc{3Par}\xspace}

\tikzset{xtick/.style={inner xsep=0pt, inner ysep=3pt, minimum
		size=0pt, draw, label=below:#1},%
	comm/.style args={#1start#2length#3color#4}{rounded corners=1mm, draw, inner
		sep=0pt, fill=#4, label=center:#1, fit={(#2,0.75)
			(#2+#3,1.5)}},%
	comp/.style args={#1start#2length#3color#4}{rounded corners=1mm, draw, inner
		sep=0pt, fill=#4, label=center:#1, fit={(#2,0)
			(#2+#3,0.75)}},%
}

\newcommand{\schedule}[3]{
	\draw[->] (-0.2, 0) -- (#1, 0) node[below] {$t$};
	\draw (0, 0) -- (0, 1.5);
	\node at (-0.8, 0.75)[rotate=90] {#2};
	\draw[dashed,gray] (0, 0.75) -- (#1, 0.75);
	\foreach \t in {0,#3} {
		\node[xtick=\t] at (\t, 0){};
	}
}
\newcommand{\task}[6][0]{
	\node[comm=#2 start #3 length #4 color #6]{};
	\node[comp=#2 start #3+#4+#1 length #5 color #6]{}; 
}

\begin{document} %
	\title{Performance Models for Data Transfers: A Case Study with Molecular Chemistry Kernels}
	\titlerunning{Performance Models for Data Transfers}
	%
	\author{Suraj Kumar \inst{1}\and 
		Lionel Eyraud-Dubois \inst{2}\and 
		Sriram Krishnamoorthy \inst{1}} %
	\authorrunning{Kumar and Eyraud-Dubois, et al.} 
	\institute{Pacific Northwest National Laboratory, Richland, Washington, USA \\
		\email{\{suraj.kumar, sriram\}@pnnl.gov}
		\and
		 Inria Bordeaux -- Sud-Ouest, Université de Bordeaux, France\\
		\email{lionel.eyraud-dubois@inria.fr}} %
	\maketitle              
	\begin{abstract} 
		With increasing complexity of hardwares, systems with different memory nodes are ubiquitous in High Performance Computing (HPC). It is paramount to develop strategies to overlap the data transfers between memory nodes with computations in order to exploit the full potential of these systems. In this article, we consider the problem of deciding the order of data transfers between two memory nodes for a set of independent tasks with the objective to minimize the makespan. We prove that with limited memory capacity, obtaining the optimal order of data transfers is a NP-complete problem. We propose several heuristics for this problem and provide details about their favorable situations. We present an analysis of our heuristics on traces, obtained by running 2 molecular chemistry kernels, namely, Hartree–Fock (HF) and Coupled Cluster Single Double (CCSD) on 10 nodes of an HPC system. Our results show that some of our heuristics achieve significant overlap for moderate memory capacities and are very close to the lower bound of makespan.

		\keywords{Communication Scheduling  \and Memory Nodes \and Runtime Systems \and Communication-Computation Overlap \and Molecular Chemistry.}
	\end{abstract} %

\section{Introduction}
\label{sec:intro}

With the advent of multicore, and the use of accelerators, it is notoriously cumbersome to exploit the full capability of a machine. Indeed, there are several challenges that come into picture. First, every architecture provides its own efficacy and interface. Therefore, a steep learning curve is required for programmers to take good utilization of all resources. Second, scheduling is a well known NP-Complete optimization problem, and hybrid and distributed resources make this problem harder (we refer ~\cite{webpagescheduling} for a survey on the complexity of scheduling problems and ~\cite{bleuse2015scheduling} for a recent survey in the case of hybrid nodes). Third, due to shared buses and parallel resources, it is challenging to obtain a precise model based on prediction of computation and communication times. Fourth, the number of architectures has increased drastically in recent years, therefore it is almost impossible to develop hand tuned optimized code for all these architectures. All these observations led to the development of different task based runtime systems. Among several runtimes, we may cite QUARK~\cite{YarKhan:2011:Quark:Manual} and PaRSEC~\cite{parsec} from from ICL, Univ. of Tennessee Knoxville (USA), StarPU~\cite{starpu} form Inria Bordeaux (France), Legion~\cite{legion12} from Stanford Univ. (USA), StarSs~\cite{ompss} from Barcelona Supercomputing Center (Spain), KAAPI~\cite{kaapi} from Inria Grenoble (France). All these runtime systems allow programmers to express their algorithms at the abstract level in the form of direct acyclic graphs (DAG), where vertices represent computations and edges represent dependencies among them. Sometimes some static information is also provided along with the DAG, such as distance to exit (last) node as a priority or affinity of computation towards resources. The runtime is then responsible for managing scheduling of computations and communications, data transfers among different memories, computation-communication overlap, and load balance.


In the last few decades, we have witnessed a drastic improvement in the hardware to provide higher rate of computation, but comparatively smaller improvement has been achieved for the rate of data movement. With extreme scale computing, supercomputers face bottlenecks due to the need of large amount of data~\cite{ascaccommitteereport2014,yelick2016}. Therefore, the HPC community is now focusing on avoiding, hiding and minimizing communication costs.

Certain applications such as dense linear algebra kernels have regular structure. Therefore, it is possible to associate priorities to computations, based on the task graph structure, and to use them at runtime to make the execution efficient. In irregular applications, programmers do not know the precise structure of the task graphs in advance: tasks are added recursively based on certain sentinel constraints. For such applications, the runtime system sees a set of independent tasks and schedules them on different processing units. It is extremely important for runtimes to decide the order of data transfers for these scheduled computations so as to maximize the overlap between computations and communications. This is the main topic of this article. We prove that the order of communications on two memory nodes with the objective of minimizing the makespan is a NP-Complete problem if the memory of the target node is limited. Our proof is inspired from work by~\cite{Papadimitriou:1980:FSL:322203.322213}, which applies a similar technique for 2-machine flowshop problem with bounded capacity. The main difference between both approaches is that they consider all tasks have the same occupation on the second machine and the memory occupation starts when the processing finishes on the first machine. On the contrary, our approach is designed for tasks appearing in scientific workloads whose memory requirements are highly irregular and we consider that memory is acquired before starting the data transfer on the communication resource. We propose different runtime strategies in order to maximize the overlap of computations and communications. We evaluate our strategies on the context of a cluster of homogeneous nodes. However, our approach is generic and easily adaptable to any system which operates on different memory spaces. Here are the important contributions of this article:

\renewcommand{\labelitemi}{$\bullet$}
\begin{itemize}
	\item NP-Completeness proof for the general data-transfer problem 
	\item Proposed different scheduling strategies with the objective to minimize the makespan
	\item Linear programming formulation of the problem
	\item Numerous experiments to assess the effectiveness of our strategies on molecular chemistry kernels 
\end{itemize}

The outline of the article is the following. Section~\ref{sec:relatedWork} describes past work on the computations with limited memory and similar problems in the literature. In section~\ref{sec:theoreticalProof}, we present an algorithm to obtain the order of data transfers when there is not any memory capacity restriction. Then, we also prove that in general data transfer problem is NP-complete. In Section~\ref{sec:heuristics}, we propose several heuristics and describe their favorable scenarios. We mainly consider three categories of heuristics: static heuristics, dynamic heuristics and static heuristics with dynamic corrections. Sections~\ref{sec:expSetting} describes our experimental setup and we evaluate our proposed strategies on two molecular chemistry kernels in Section~\ref{sec:expResults}. Our results show that static heuristics with dynamic corrections achieve good performance in most cases. We finally propose conclusions and perspectives in Section~\ref{sec:conclusion}.

\section{Related Work}
\label{sec:relatedWork}

Historically there has been a great emphasis on the development of parallel algorithms and minimizing the complexity of computations. As the number of computation cores has increased drastically in recent years, supercomputers face bottleneck due to communication required by an application. Hence, in recent years the focus has changed towards developing communication avoiding algorithms, strategies to hiding communications and minimizing the data accessed by applications~\cite{yelick2016}.

The problem of scheduling tasks has been highly studied in the literature and many formulations are known to be NP-Complete~\cite{GareyJohnson}. Many static and dynamic strategies have been proposed and analyzed for scheduling tasks on heterogeneous resources~\cite{heft-Topcuoglu,hipc16multiresource,ipdps16starpu}. There is also a number of studies in the direction of task scheduling with the emphasis on improving locality and minimizing the communication cost~\cite{starpu,heft-Topcuoglu}. Stanisic et. al~\cite{luka-dmdar} proposed a heuristic to schedule tasks on a computational resource where most of its data is available. A similar approach has been adopted by Agullo et. al for the scheduling of sparse linear algebra kernels~\cite{agullo_fmm}. Predari et. al proposed heuristics to partition the task graph across a number of processors such that inter-processor communication can be minimized~\cite{predari:tel-01518956}.

The problem considered in this article also can be viewed as a flow shop problem: the communication link can be seen as a processing resource, and each task needs to first be handled by the communication link and then by the computational resource. Communication and computation times of a task can thus be considered as processing times on different machines. Johnson has provided scheduling strategies for 2 and 3-machine flow shop problems with infinite memory capacity~\cite{johnson}. 2-machine flow shop problem with finite buffer has been proven NP-Complete by Papadimitriou et. al~\cite{Papadimitriou:1980:FSL:322203.322213}, in which a constraint is imposed on the number of tasks that can await execution on the second machine.

A number of other studies have focused on scheduling with limited memory and storage, starting with the work of register allocation for arithmetic expressions by Sethi and Ulman~\cite{Sethi:1970:GOC:321607.321620}. Sarkar et. al worked on the scheduling of graphs of smaller-grain tasks with limited memory, where each task requires homogeneous data size~\cite{vsarkar-pact}. The same work has been extended by Marchal et. al for task graphs where memory requirement of each task is highly irregular~\cite{loris-ipdps18}.

\section{Data Transfer Problem Formulation}
\label{sec:theoreticalProof}

To exploit the full potential of a system it may be necessary to execute tasks on processing
units where all of their data does not reside. A task may require all of its input data
in local memory before starting the computation. There may be multiple tasks
scheduled on a processing unit, which require to transfer data from the same
memory node. Ordering data transfers for such tasks is very crucial for the
communication-computation overlap, thus for the overall
performance. In general, order of task execution with input
and output data transfers can be  
viewed as a 3-machine flowshop problem, where processing time on the first machine is 
input data transfer time, processing time on the second machine is task computation time, 
and processing time on the third machine is output data transfer time; and the objective is 
to minimize the total makespan. This is a well known NP-complete problem~\cite{3machineFlowShopNPComplete}.

In many cases, output data that needs to be retrieved after task execution is much smaller than the input data. It is often the case that future tasks running on the same memory node require output data of the past tasks. Therefore, most runtime systems transfer data to other memory nodes based on the demand -- not immediately after they were produced. It is also possible that all output data can be stored in a preallocated separate buffer on a memory node. Hence, we do not consider output data explicitely in our analysis and assume that output data is negligible or stored in a separate buffer for each task. Thus problem considered here is more similar to a 2-machine flowshop problem. We prove that ordering the execution of such tasks with finite memory capacity is a NP-complete problem: 


\noindent\textbf{Problem $DT$} : A set of tasks $ST=\{T_1,
\cdots, T_n\}$ is scheduled on a processing unit $P$ with
memory unit $M$ of capacity $C$. Input data for tasks of $ST$
reside on another memory unit $M'$. $CM_i$ is the communication time to
transfer input data from $M'$ to $M$ for task $i$ and $CP_i$
is the computation time of task $i$ on $P$. We assume that these
tasks do not produce any output data. There can be only one
communication at a time, and $P$ can only process one task at
a time. A task uses an amount of memory in  $M$ from the
start of its communication to the end of its computation.

\noindent Given $L$, is there a feasible schedule $S$ for $ST$ such that
makespan of $S$, $\mu(S) \le L$?

Given a schedule, \scomm($i$) and \scomp($i$) represent the
start times of task $i$ on communication and computation
resources. A schedule is feasible if for every time $t$, the
amount of memory required by all tasks such that $\scomm(i) \leq t
\leq \scomp(i) + CP_i$ is not more than the memory capacity
$C$. 
For simplicity, we assume throughout the article that tasks
require memory only to store their input data, and thus that
the amount of memory required by a task is proportional to its
communication time. Without loss of generality, we consider in
all examples of Sections~\ref{sec:theoreticalProof}
and~\ref{sec:heuristics} that the memory requirement of a task
is equal to its communication time.

We call a task $i$ compute intensive if $CP_i \ge CM_i$, and
communication intensive otherwise.

\subsection{Special case: Infinite Memory}

When the computational resource has a very large memory, our problem
becomes a classic 2-machine flowshop problem: communication time is
the processing time on the first machine and computation time is the
processing time on the second machine. Johnson's
algorithm~\cite{johnson} is known to provide an ordering for the tasks
which results in an optimal makespan. This algorithm is rewritten in
Algorithm~\ref{alg:OrderOfExecutionInfinteMemory}.

\begin{algorithm}
	\caption{\label{alg:OrderOfExecutionInfinteMemory}Johnson's~\cite{johnson} algorithm (infinite memory case).}
	\begin{algorithmic}[1]
		\STATE Divide ready tasks in two sets $S_1$ and $S_2$. If computation time of a task $T$ is not less than its communication time, then $T$ is in $S_1$ otherwise in $S_2$.
		\STATE Sort $S_1$ in queue $Q$ by non-decreasing communication times
		\STATE Sort $S_2$ in queue $Q'$ by non-increasing computation times
		\STATE Append $Q'$ to $Q$
		\STATE $\tau_{\text{COMM}} \gets 0$ \hfill\COMMENT{Available time of communication resource}
		\STATE $\tau_{\text{COMP}} \gets 0$\hfill \COMMENT{Available time of computation resource}
		\WHILE{$Q \neq \emptyset$}
		\STATE Remove a task $T$ from beginning of $Q$ for processing
		\STATE $\scomm(T) \gets \tau_{\text{COMM}}$
		\STATE $\scomp(T) \gets max(\scomm(T) + CM_T, \tau_{\text{COMP}})$
		\STATE $\tau_{\text{COMM}} \gets \scomm(T) + CM_T$
		\STATE $\tau_{\text{COMP}} \gets \scomp(T) + CP_T$
		\ENDWHILE
	\end{algorithmic}
\end{algorithm}

We  prove optimality of Algorithm~\ref{alg:OrderOfExecutionInfinteMemory}  differently. Our proof rely on the following lemma.

\begin{lemma}\label{lemma:swappingOfTasks}
	Swapping two contiguous tasks $A$ and $B$ in a schedule does not improve the makespan if one of the conditions is true.
	\begin{enumerate}[label=\roman*)]
		\item  $CP_A \ge CM_A, CP_B \ge CM_B, CM_A \le CM_B$
		\item $CP_A < CM_A, CP_B < CM_B, CP_A \ge CP_B$
		\item $CP_A \ge CM_A, CP_B < CM_B$
	\end{enumerate}
\end{lemma}
\begin{proof}
	Let $t_1$ and $t_2$ be the early start time on communication and computation resources just before the task $A$ starts.
	
	\begin{figure}[htb]
		\centering
		\includegraphics[scale=0.35]{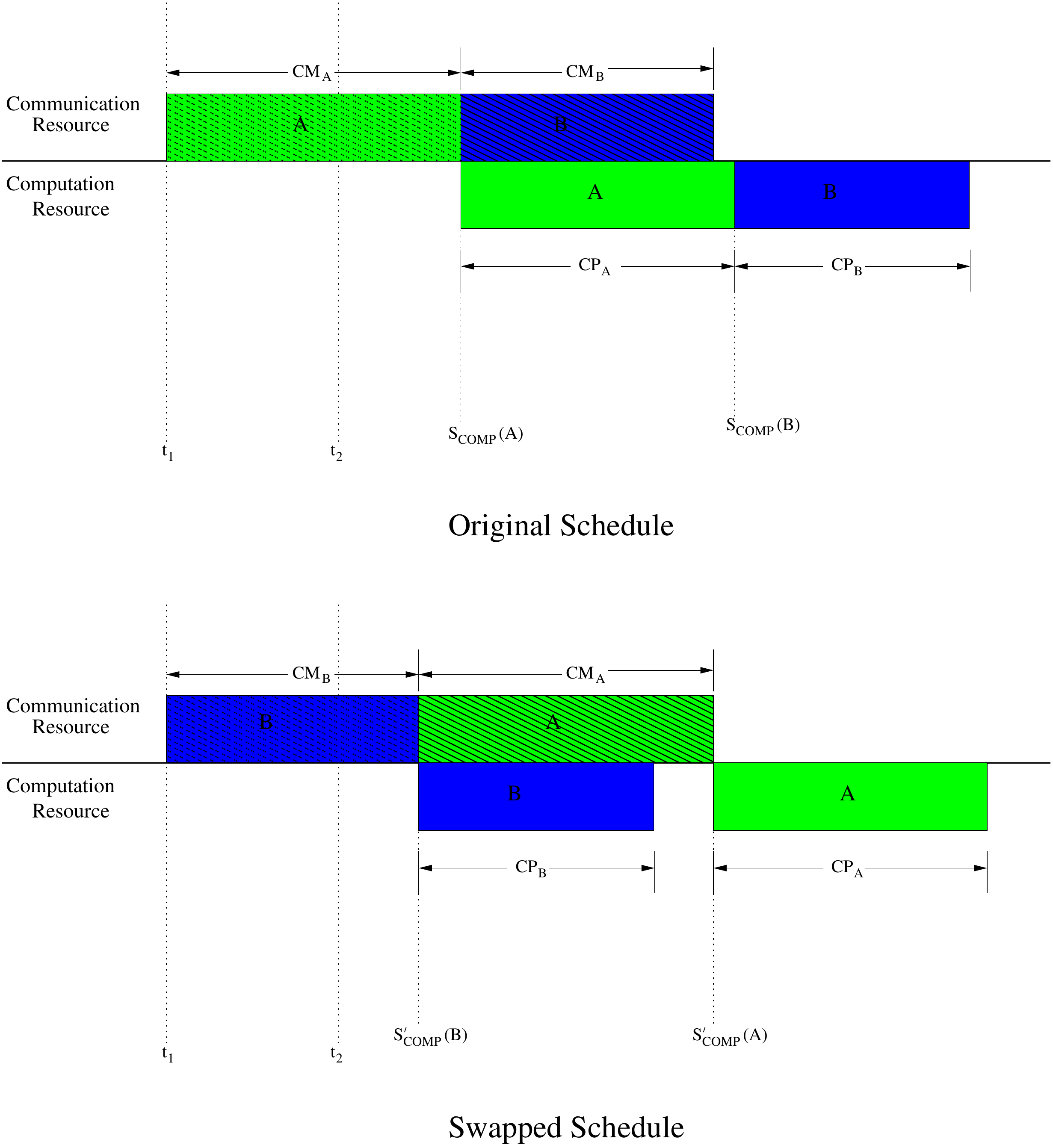}
		\caption{ \label{fig:bothSchedule} Original and Swapped Schedules.}
	\end{figure}
	We write the following constraints based on the two schedules of the Fig~\ref{fig:bothSchedule}.
	
	\noindent $\scomp(A) = max(t_1 + CM_A, t_2)$ \\
	\noindent $\scomp(B) = max(\scomp(A) + CP_A,  t_1 + CM_A + CM_B)$\\
	
	\noindent Completion time of $B$ in original Schedule = $\scomm(B) + CP_B$\\
	\noindent Completion time of $B$ in swapped Schedule,  $\scomp '(B)  =  max(t_1 + CM_B, t_2)$\\
	\noindent Completion time of $A$ in swapped Schedule, $ \scomp '(A)$ \newline
	 \hspace*{0pt}\hfill $ = max(\scomp'(B) + CP_B, t_1 + CM_B + CM_A)$

	\noindent In both schedules, early available time on communication resource after scheduling $A$ and $B$ is same. If we show that early available time on computation resource in swapped schedule after scheduling both tasks is not less than the time of original schedule, then the proof is complete. Hence our goal is to prove that,		
	$$\scomp(B) + CP_B \le \scomp'(A) + CP_A $$
	
{\small
		\begin{align*}
				\scomp(B) + CP_B & =  &max(\scomp(A) + CP_A + CP_B,  t_1 + CM_A + CM_B + CP_B)\\
				& = & max(t_1 + CM_A + CP_A + CP_B, t_2 + CP_A + CP_B, t_1 + CM_A\\
				& &+ CM_B + CP_B) 
		\end{align*}
	Case I: $CP_A \ge CM_A, CP_B \ge CM_B, CM_A \le CM_B$
	\begin{align*}
		\scomp(B) + CP_B & =  max(t_1 + CM_A + CP_A + CP_B, t_2 + CP_A + CP_B, t_1 + CM_A \\
		&\pushright{+ CM_B + CP_B)}\\
		& \le max(t_1 + CM_A + CP_B, t_2  + CP_B, t_1  + CM_B + CP_B) + CP_A\\
		& =  max( t_1  + CM_B + CP_B, t_2  + CP_B) + CP_A \\
		& =  max(max(t_1 + CM_B, t_2) + CP_B, t_1 + CM_A + CM_B) + CP_A\\
		& =  max(\scomp'(B) + CP_B, t_1 + CM_A + CM_B) + CP_A \\
		& =  \scomp'(A) + CP_A
	\end{align*}
	Case II: $CP_A < CM_A, CP_B < CM_B, CP_A \ge CP_B$
	\begin{align*}
		\scomp(B) + CP_B & =  max(t_1 + CM_A + CP_A + CP_B, t_2 + CP_A + CP_B, t_1 + CM_A \\
		& \pushright{+ CM_B + CP_B)}\\
		& \le  max(t_1 + CM_A + CP_B, t_2+CP_B, t_1 + CM_A + CM_B) + CP_A\\
		& =  max(t_1 + CM_A + CM_B, t_2+ CP_B) + CP_A \\
		& =  max(t_1 + CM_B + CP_B, t_2 + CP_B, t_1 + CM(A) + CM_B) + CP_A \\
		& =  max(max(t_1 + CM_B, t_2) + CP_B, t_1 + CM_A + CM_B) + CP_A \\
		& =  max(\scomp'(B) + CP_B, t_1 + CM_A +CM_B) + CP_A \\
		& =  \scomp'(A) + CP_A
	\end{align*}
	Case III: $CP_A \ge CM_A, CP_B < CM_B$
	\begin{align*}
	\scomp(B) + CP_B & =  max(t_1 + CM_A + CP_A + CP_B, t_2 + CP_A + CP_B, t_1 + CM_A\\ 
	&\pushright{+ CM_B + CP_B)}\\
	& \le  max(t_1 + CM_A + CP_B, t_2 + CP_B, t_1 + CM_B + CP_B) + CP_A\\
	& \le  max(t_1 + CM_A + CM_B, t_2 + CP_B, t_1 + CM_B + CP_B) + CP_A \\
	& =  max(max(t_1 + CM_B, t_2) + CP_B, t_1 + CM_A + CM_B) + CP_A \\
	& =  max(\scomp'(B) + CP_B, t_1 + CM_A +CM_B) + CP_A \\
	& =  \scomp'(A) + CP_A
\end{align*}
}

\end{proof}

\begin{theorem}
	Scheduled constructed by Algorithm~\ref{alg:OrderOfExecutionInfinteMemory} achieves optimal makespan.
\end{theorem}

\begin{proof}
	
	Let $O$ be an optimal schedule. We assume that $O$ is a permuation schedule, If it is not the case, we make the order of computations same as the order of communications without increasing the makespan.  Suppose two tasks have opposite order on both resources then we position the second task just before the first task on the computation resource. It is evident that this change does not alter the optimal  makespan, we repeat this procedure until order of communications and computations is same in $O$.
	
	Let $S$ be the  schedule obtained from Algorithm~\ref{alg:OrderOfExecutionInfinteMemory}. We prove the theorem by converting $S$ to $O$ and showing that at each step makespan of intermediate schedule is not less than the original makespan. We rely on Lemma~\ref{lemma:swappingOfTasks} to convert $S$ to $O$.
	
	We traverse schedule $O$ from left to right and for each $ith$ task in sequence, we apply Lemma~\ref{lemma:swappingOfTasks} repetitively until order of task in the swapped schedule is same. It is obvious that after moving the $ith$ task at the beginning remaining schedule (schedule after $ith$ task) still satisfies one of the conditions of Lemma~\ref{lemma:swappingOfTasks}.
	
	Let the final swapped schedule is $S_{final}$. From Lemma~\ref{lemma:swappingOfTasks}, $makespan(S_{final})$ $ \ge $ $makesapn(S)$. From the construction, $makespan(S_{final})$ $ \le$ $ makesapn(O)$. As $O$ is an optimal schedule, hence $makespan(S) = makesapn(O)$. This completes our proof.
\end{proof}

\subsection{Finite Memory}

We now consider the general case, in which the memory limit is a
constraint for the schedule. This is related to previous work by
Papadimitriou et. al~\cite{Papadimitriou:1980:FSL:322203.322213}, in which the second machine
can only handle a bounded number of tasks. Our problem generalizes this
work to heterogeneous memory consumption among tasks, with an
additional difference: memory usage starts at the beginning of the
first part of a task (instead of at the end of the first part). This
requires to provide a slightly different NP-completeness proof, as
given below.

\begin{theorem}\label{th:npComplete}
	Problem $DT$ is NP-complete.
\end{theorem}
\begin{proof}
	
	It is easy to see that the $DT$ belongs in NP: given a schedule, one
	can check in linear time that at each start of a communication, the
	memory constraint is satisfied, and that task starts computation only
	after its input data is transferred to $M$.
	
	In order to prove NP-hardness, we use a reduction from the
	well-known NP-complete problem 3 Partition~\cite{GareyJohnson}: 
	
	\noindent \textbf{Three Partition Problem} (\threepart): Given a set of
	$3m$ integers $A = \{ a_1, \cdots, a_{3m }\}$, is there a partition
	of $A$ into $m$ triplets $TR_i = \{a_{i_1}, a_{i_2}, a_{i_3}\}$,
	such that $\forall i, a_{i_1} + a_{i_2} + a_{i_3} = b$, where
	$b=(1/m) \sum a_i $?
	
	Let us first show that \threepart problem reduces in polynomial time to
	problem $DT$. Suppose that we are given an instance $A = \{ a_1,
	\cdots, a_{3m }\}$ of \threepart. It is immediately obvious that
	$a_i>1$, since we can always add sufficiently large integers to
	the $a_i$ values and scale the problem accordingly. This scaling will not
	affect in any way the existence of a solution for the instance of
	\threepart problem.
	
	From such an instance, we define $x = max\{a_i:1\le i\le 3m\}$,
	and we construct an instance $I$ of the problem $DT$ with $4m+1$
	tasks, whose characteristics are given in
	Table~\ref{table:np.completeness.tasks}.
	
	\begin{table}[htb]
		\begin{tabular}{ |c|c|c| }
			\hline
			Task & Communication time & Computation time \\ \hline
			$K_0$ & $0$ & $3$ \\ \hline
			$K_1, \cdots, K_{m-1}$ & $b'=b+6x$ & $3$\\ \hline
			$K_m$ & $b'=b+6x$ & $0$ \\ \hline
			$1\le i \le 3m, A_i$ & $1$ & $a_i' = a_i + 2x$\\ \hline
		\end{tabular}
		
		\noindent Memory capacity: $C=b'+3$\\
		\noindent Target makespan: $L=m(b'+3)$
		\caption{Definition of tasks in the reduction from \threepart.}
		\label{table:np.completeness.tasks}
	\end{table}
	
	As mentioned previously, we consider memory requirement of a task is equal to its communication time. If it is not the case, we can adjust $C$ such that at any point in a schedule at max one $K_i$ and three $A_i$ tasks can be active.

	We show that $I$ has a schedule $S$ with makespan at most $L$ if and
	only if the original \threepart instance has a solution. Notice that
	the sum of communication times and the sum of computation times are both
	equal to $L$, therefore a valid schedule of makespan at most $L$ has
	makespan exactly $L$, with no idle time on both resources. It
	indicates that the first task is $K_0$ and the last task is $K_m$.

	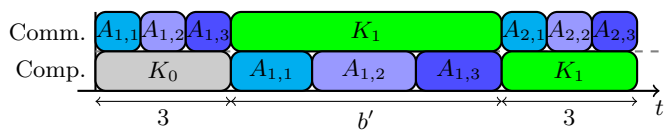
\begin{figure}[htb]
		\tikzset{xtick/.style={inner xsep=0pt, inner ysep=3pt, minimum size=0pt, draw},%
			task/.style args={#1start#2length#3res#4color#5}{rounded corners, draw, inner
				sep=0pt, fill=#5, label=center:#1, fit={(#2,#4*0.75) (#2+#3,#4*0.75+0.75)}},%
			vert/.style={inner sep=1pt, fill=black, circle, draw, label=#1}
		}
		\newcommand{\scheduleNoName}[1]{
			\draw[->] (-0.4, 0) -- (#1, 0) node[below] {$t$};
			\draw (0, 0) -- (0, 1.5) node[pos=0.25, left] {Comp.}
			node[pos=0.75, left] {Comm.};
			\draw[dashed,gray] (0, 0.75) -- (#1, 0.75);
		}
		\centering
		\begin{tikzpicture}[yscale=0.7, thick, xscale=0.6]
		\scheduleNoName{12.5}
		\node[task=$A_{1,1}$ start 0 length 1 res 1 color cyan]{};
		\node[task=$A_{1,2}$ start 1 length 1 res 1 color blue!40!white]{};
		\node[task=$A_{1,3}$ start 2 length 1 res 1 color blue!70!white]{};
		\node[task=$K_0$ start 0 length 3 res 0 color gray!40!white]{};
		\node[task=$K_1$ start 3 length 6 res 1 color green]{}; 
		\node[task=$A_{1,1}$ start 3 length 1.8 res 0 color cyan]{};
		\node[task=$A_{1,2}$ start 4.8 length 2.3 res 0 color blue!40!white]{};
		\node[task=$A_{1,3}$ start 7.1 length 1.9 res 0 color blue!70!white]{};
		\node[task=$A_{2,1}$ start 9 length 1 res 1 color cyan]{};
		\node[task=$A_{2,2}$ start 10 length 1 res 1 color blue!40!white]{};
		\node[task=$A_{2,3}$ start 11 length 1 res 1 color blue!70!white]{};
		\node[task=$K_1$ start 9 length 3 res 0 color green]{};
		\draw[<->,thin] (0, -0.2) -- node[below]{$3$} (3, -0.2) ;
		\draw[<->,thin] (9, -0.2) -- node[below]{$3$} (12, -0.2) ;
		\draw[<->,thin] (3, -0.2) -- node[below]{$b'$} (9, -0.2) ;
		\end{tikzpicture}
		\caption{ \label{fig:firstSegment} Pattern of feasible schedule $S$.}
	\end{figure}

	If the \threepart instance has a solution, $A$ can be partitioned into
	$m$ triplets $TR_i = \{a_{i_1}, a_{i_2}, a_{i_3}\}$ such that $\forall
	i, a_{i_1} + a_{i_2} + a_{i_3} = b$, then we can construct a feasible
	schedule $S$ without idle times by the pattern depicted in
	Figure~\ref{fig:firstSegment}. The communications of tasks in $TR_i$ take
	place during the computation of task $K_{i-1}$, and the computations
	of tasks in $TR_i$ take place during the communication of task
	$K_i$. Since the memory capacity is $C=b'+3$, all tasks from a triplet
	can fit in memory with a task $K_i$, and their durations are exactly
	equal to the communication time of $K_i$. This schedule is thus
	feasible, and has length exactly $L$.
	
	\medskip
	
	We now prove that any feasible schedule of $I$ corresponds to a valid
	decomposition of $A$ for \threepart. Indeed, we argue that every
	feasible schedule has to consist of $m$ segments like the one shown in
	Figure~\ref{fig:firstSegment}. Each segment provides a triplet
	$\{a_{i_1}, a_{i_2}, a_{i_3}\}$ such that $a_{i_1} + a_{i_2} + a_{i_3}
	= b$.
	
	Any schedule $S$ of $I$ having no idle time must start with $K_0$. We
	first show that no other $K_i$ task can be active with $K_0$,
	otherwise we would get idle time on computation resource. Indeed, the
	communication of such a task $K_i$ would end at time at least $b'>3 +
	6x$, but at most two $A_i$ tasks can be computed, and they end at time
	at most $3+2max\{a_i':1\le i\le 3m\} = 3 + 6x$
	
	Hence three $A_i$ tasks must follow $K_0$. The memory requirement of
	other $K_i$ tasks is $b'$ and $2b'>C$, therefore at any point in the
	schedule at most one $K_i$ task can be active. Since the total
	duration of all $K_i$ tasks is $3 + (m-1)(b'+3) + b' = m (b'+3)=L$, at
	each point in $S$ exactly one $K_i$ task is active.
	
	With these $K_i$ tasks in place, the schedule on the computation
	resource contains $m$ slots of length exactly $b'$, in which all $A_i$
	tasks must fit without preemption. We can thus define triplet $TR_i$
	as the set of tasks which execute during the communication phase of
	task $K_i$. Since at each point in $S$, exactly one $K_i$ task is
	active, and since $S$ has no idle time on the computation resource,
	the total computation time of tasks in $TR_i$ is exactly $b'$, and thus
	$a_{i_1} + a_{i_2} + a_{i_3} = b$. This partition is thus a valid
	solution for the \threepart instance $A$. 
	
\end{proof}


This theorem shows that adding a memory constraint to our
problem makes it more difficult. One additional difficulty
compared to infinite memory capacity 2-machine flowshop~\cite{johnson} is that
it may not be optimal to consider the same ordering on both
machines:

	\begin{proposition}
	There exists an instance of $DT$ for which in all optimal
	schedules, the communication order of tasks is different
	from their computation order.
\end{proposition}

\begin{table}
	\begin{center}
		\begin{tabular}{|c|c|c|c|}
			\hline
			\multirow{2}{*}{Task} & Memory Req & \multirow{2}{*}{Comm Time} & \multirow{2}{*}{Comp Time}\\  
			&=Comm Volume && \\ \hline
			A & 0 & 0 & 5\\ \hline
			B & 4 & 4 & 3\\ \hline
			C & 1 & 1 & 6\\ \hline
			D & 3 & 3 & 7\\ \hline
			E & 6 & 6 & 0.5\\ \hline
			F & 7 & 7 & 0.5\\ \hline
		\end{tabular}
	\end{center}
	\caption{Example instance where ordering on both resources
		has to be different.}
	\label{table:different.order}
\end{table}

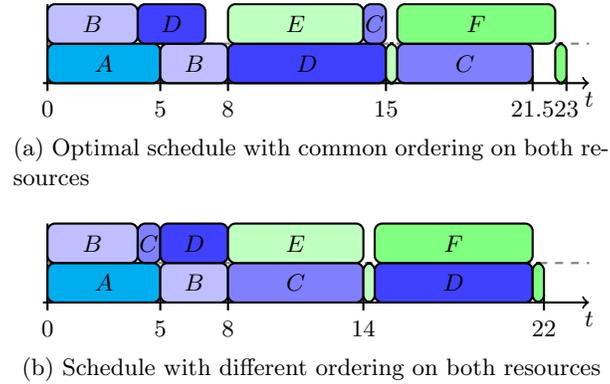
\begin{figure}
	\centering
	\newcommand{\taskA}[1]{\node[comp=$A$ start #1 length 5 color cyan] {};}
	\newcommand{\taskB}[2][0]{\task[#1]{$B$}{#2}{4}{3}{blue!25!white}}
	\newcommand{\taskC}[2][0]{\task[#1]{$C$}{#2}{1}{6}{blue!50!white}}
	\newcommand{\taskD}[2][0]{\task[#1]{$D$}{#2}{3}{7}{blue!75!white}}
	\newcommand{\taskE}[2][0]{\node[comm=$E$ start #2 length 6 color green!25!white]{};%
		\node[comp={} start #2+6+#1 length 0.5 color green!25!white]{};}
	\newcommand{\taskF}[2][0]{\node[comm=$F$ start #2 length 7 color green!50!white]{};%
		\node[comp={} start #2+7+#1 length 0.5 color green!50!white]{};}
	
	\subfloat[][Optimal schedule with common ordering on both
	resources]{
		\begin{tikzpicture}[yscale=0.7, thick, xscale=0.3]
		\schedule{24}{}{5,8,15,21.5,23}
		\taskA{0}
		\taskB[1]{0}
		\taskD[1]{4}
		\taskE[1]{8}
		\taskC[0.5]{14}
		\taskF{15.5}
		\end{tikzpicture}
		\label{fig:different.order.same}
	}
	
	\subfloat[][Schedule with different ordering on both
	resources]{
		\begin{tikzpicture}[yscale=0.7, thick, xscale=0.3]
		\schedule{24}{}{5,8,14,22}
		\taskA{0}
		\taskB[1]{0}
		\taskC[3]{4}
		\taskD[6.5]{5}
		\taskE{8}
		\taskF{14.5}
		\end{tikzpicture}
		\label{fig:different.order.opt}
	}
	\caption{Schedules for the instance of
		Table~\ref{table:different.order} with a memory capacity of 10.}
	\label{fig:different.order}
\end{figure}

\begin{proof}
	Consider the instance described on
	Table~\ref{table:different.order}, in which memory capacity
	is $C=10$. Figure~\ref{fig:different.order.same} shows the
	best possible schedule when tasks are scheduled in the same
	order on both resources (obtained by exhaustive search). On
	the other hand, Figure~\ref{fig:different.order.opt} shows
	another schedule with lower makespan, in which the order is
	different.
	
	In the infinite memory case, the standard proof that an
	optimal schedule exists with the same order on both
	resources claims that it is possible to swap two tasks which
	do not satisfy this property. On
	Figure~\ref{fig:different.order}, this would mean swapping
	tasks $D$ and $E$. But the communication of task $E$ can not
	start earlier because it would not fit in memory with tasks
	$B$ and $C$, and delaying the computation of task $E$ after
	task $D$ would delay task $F$ because $E$ and $F$ do not fit
	in memory together. We can see that this claim does not hold
	in the constrained memory case. 
\end{proof}

\section{Data Transfer Order heuristics}
\label{sec:heuristics}

Algorithm~\ref{alg:OrderOfExecutionInfinteMemory} presented in Section~\ref{sec:theoreticalProof} achieves an optimal makespan when there is no memory constraint. This optimal value indicates a lower bound on the makespan of the constrained case. We denote this value with \textit{optimal makespan infinite memory} ($OMIM$). In the present Section, we propose different heuristics for the limited memory case, and we assess their efficiency with respect to this lower bound in Section~\ref{sec:expResults}.

We classify our heuristics into mainly three categories. In the first category, the order of all computations and communications is computed in advance and the same order is followed on both resources. In the second category, the next task to schedule is dyanmically chosen based on different criteria. The final category is based on combining strategies from the first two categories: a static ordering is precomputed and corrected dynamically to avoid idle time caused by memory limitations. In all of our strategies (except linear programming based strategy), communication and computations take place in the same order.

\subsection{Static Ordering}
	In this class of strategies, we compute the order of processing in advance based on criteria such as communication time and computation time. After computing the order, we follow the same sequence on computation and communication resources and make sure that the memory constraint is respected at each point in the schedule.

In Algorithm~\ref{alg:OrderOfExecutionInfinteMemory}, compute intensive tasks are sorted in increasing order of communication times. It allows tasks to utilize the computation resource maximally and make enough margin on the communication resource to accommodate more communication intensive tasks with maximum overlap. Communication intensive tasks are sorted in decreasing order of computation time, which allows tasks to utilize the margin created on communication resource. Hence, in this section, we obtain the orders by sorting tasks based on different combinations of communication and computation times.

\begin{enumerate}[label=\roman*)]
	\item \textit{order  of  optimal strategy infinite memory} ($OOSIM$): This heuristic uses the order given by Algorithm~\ref{alg:OrderOfExecutionInfinteMemory}, but respects the memory constraint at each point in the schedule. Hence the makespan of this heuristic may be completely different from $OMIM$.
	
	\item \textit{increasing order of communication strategy} ($IOCMS$): Tasks are ordered  in non-decreasing order of communication time. 
	
	\item \textit{decreasing order of computation strategy} ($DOCPS$): Tasks are ordered in non-increasing order of computation time. 
	\item \textit{increasing order of communication plus computation strategy} ($IOCCS$): Tasks are ordered in non-decreasing order of the sum of their communication and computation times.
	\item \textit{decreasing order of communication plus computation strategy} ($DOCCS$): Tasks are ordered in non-increasing order of the sum of their communication and computation times.
	
\end{enumerate}
\begin{table}[htb]
	\begin{center}
		
		\begin{tabular}{|c|c|c|c|}
			\hline
			\multirow{2}{*}{Task} & Memory Req & \multirow{2}{*}{Comm Time} & \multirow{2}{*}{Comp Time}\\  
			&=Comm Volume && \\ \hline
			A & 3 & 3 &  2\\ \hline
			B & 1 & 1 & 3\\ \hline
			C & 4 & 4 & 4\\ \hline
			D & 2 & 2 & 1\\ \hline
		\end{tabular}
		\caption{\label{tab:staticOrderExample} A task set for static order schedules.}
	\end{center}
\end{table}

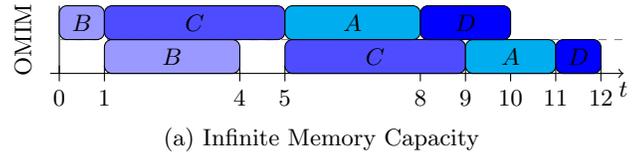
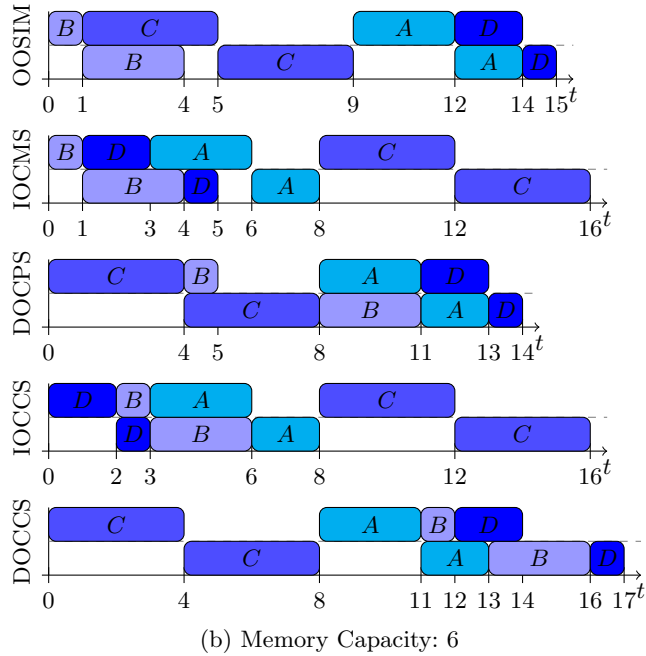
\begin{figure}[htb]
	\newcommand{\taskA}[2][0]{\task[#1]{$A$}{#2}{3}{2}{cyan}}
	\newcommand{\taskB}[2][0]{\task[#1]{$B$}{#2}{1}{3}{blue!40!white}}
	\newcommand{\taskC}[2][0]{\task[#1]{$C$}{#2}{4}{4}{blue!70!white}}
	\newcommand{\taskD}[2][0]{\task[#1]{$D$}{#2}{2}{1}{blue}}
	
	\centering
	\subfloat[][Infinite Memory Capacity]{
		\begin{tikzpicture}[scale=0.6]
		\schedule{12.5}{OMIM}{1,4,5,8,9,10,11,12}
		\taskB{0}
		\taskC{1}
		\taskA[1]{5}
		\taskD[1]{8}
		\end{tikzpicture}
	}
	
	\subfloat[][Memory Capacity: 6]{
		\begin{tikzpicture}[yscale=0.6, xscale=0.45]
		\begin{scope}
		\schedule{15.5}{OOSIM}{1,4,5,9,12,14,15}
		\taskB{0}
		\taskC{1}
		\taskA{9}
		\taskD{12}
		\end{scope}
		\begin{scope}[yshift=-2.75cm]
		\schedule{16.5}{IOCMS}{1,3,4,5,6,8,12,16}
		\taskB{0}
		\taskD[1]{1}
		\taskA{3}
		\taskC{8}
		\end{scope}
		\begin{scope}[yshift=-5.5cm]
		\schedule{14.5}{DOCPS}{4,5,8,11,13,14}
		\taskC{0}
		\taskB[3]{4}
		\taskA{8}
		\taskD{11}
		\end{scope}
		\begin{scope}[yshift=-8.25cm]
		\schedule{16.5}{IOCCS}{2,3,6,8,12,16}
		\taskD{0}
		\taskB{2}
		\taskA{3}
		\taskC{8}
		\end{scope}
		\begin{scope}[yshift=-11cm]
		\schedule{17.5}{DOCCS}{4,8,11,12,13,14,16,17}
		\taskC{0}
		\taskA{8}
		\taskB[1]{11}
		\taskD[2]{12}
		\end{scope}
		\end{tikzpicture}
	}
	
	\caption{ \label{fig:staticOrderExample} Different static order heuristic schedules for Table~\ref{tab:staticOrderExample}}. 
\end{figure}


In order to highlight the different behaviors of these static strategies, we propose on Table~\ref{tab:staticOrderExample} an example instance, and on Figure~\ref{fig:staticOrderExample} the corresponding schedules for all these heuristics.

\subsection{Dynamic Selection}
Dynamic strategies are based on the following principle: when the communication resource is idle, a task is chosen based on a selection criterion which differs depending on the heuristic, among those which fit in memory and induce minimum idle time on the computation resource.  For example, if the selection criterion is to choose a highly compute intensive task, then we compute the ratio of computation time and communication time for all tasks, and we select a task with the maximum ratio among those which induce minimum idle time on the computation resource and fit in the currently available memory. If no task fits in memory then we leave the resource idle at that point and proceed to the next event point. We also ensure that the order on both resources is the same, by ordering tasks on the computation resource accordingly.



\begin{enumerate}[label=\roman*)]
	\item \textit{largest communication task respects memory restriction} ($LCMR$): A task with the largest communication time is chosen. 
	\item \textit{smallest communication task respects memory restriction} ($SCMR$): A task with the smallest communication time is chosen.
	\item \textit{maximum accelerated task respects memory restriction} ($MAMR$): A task with the maximum ratio of computation time to communication time is chosen.
\end{enumerate}
\begin{table}[htb]
	\begin{center}
		
		\begin{tabular}{|c|c|c|c|}
			\hline
			\multirow{2}{*}{Task} & Memory Req & \multirow{2}{*}{Comm Time} & \multirow{2}{*}{Comp Time}\\  
			&=Comm Volume && \\ \hline
			A & 3 & 3 & 2\\ \hline
			B & 1 & 1 &  6\\ \hline
			C & 4 & 4 & 6\\ \hline
			D & 5 & 5 & 1\\ \hline
		\end{tabular}
		\caption{\label{tab:dynamicSelectionExample} A task set for dynamic schedules.}
	\end{center}
\end{table}

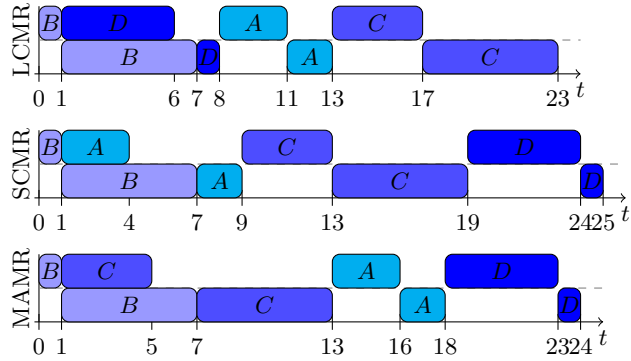
\begin{figure}[htb]
	\centering
	
	\newcommand{\taskA}[2][0]{\task[#1]{$A$}{#2}{3}{2}{cyan}}
	\newcommand{\taskB}[2][0]{\task[#1]{$B$}{#2}{1}{6}{blue!40!white}}
	\newcommand{\taskC}[2][0]{\task[#1]{$C$}{#2}{4}{6}{blue!70!white}}
	\newcommand{\taskD}[2][0]{\task[#1]{$D$}{#2}{5}{1}{blue}}
	
	\centering
	\begin{tikzpicture}[yscale=0.6,xscale=0.3]
	\begin{scope}
	\schedule{24}{LCMR}{1,6,7,8,11,13,17,23}
	\taskB{0}
	\taskD[1]{1}
	\taskA{8}
	\taskC{13}
	\end{scope}
	\begin{scope}[yshift=-2.75cm]
	\schedule{26}{SCMR}{1,4,7,9,13,19,24,25}
	\taskB{0}
	\taskA[3]{1}
	\taskC{9}
	\taskD{19}
	\end{scope}
	\begin{scope}[yshift=-5.5cm]
	\schedule{25}{MAMR}{1,5,7,13,16,18,23,24}
	\taskB{0}
	\taskC[2]{1}
	\taskA{13}
	\taskD{18}
	\end{scope}
	\end{tikzpicture}
	\caption{ \label{fig:dynamicSelectionExample} Different dynamic heuristic schedules for a task set of Table~\ref{tab:dynamicSelectionExample} with a memory capacity of 6.}
\end{figure}

We highlight the different dynamic heuristics with the instance described on Table~\ref{tab:dynamicSelectionExample} (these heuristics are too similar on the instance from the previous class), and Figure~\ref{fig:dynamicSelectionExample} shows the corresponding schedules.

\subsection{Static Order with Dynamic Corrections}

In this class of strategy, we precompute the order of tasks based on some criterion and then follow this ordering as much as possible. But when the communication resource is idle because the memory requirement of the next task is too high, then we select a task  with a dynamic strategy. After a task is dynamically selected, we update the remaining order without this task. This class of strategy takes advantage of static information in the form of precomputed order and dynamic corrections to minimize the idle time due to memory constraint.

In all strategies of this class, the initial order is $OMIM$ order,  obtained by Algorithm~\ref{alg:OrderOfExecutionInfinteMemory}. We define the following heuristics based on how we select a task from the set of tasks which fit in memory and induce minimum idle time on the computation resource. If no task fits in memory then we leave the resource idle and forward to the next event point.

\begin{enumerate}[label=\roman*)]
	\item \textit{optimal order infinite memory largest communication task respects memory restriction} ($OOLCMR$): A task with the largest communication time is chosen from the set.
	\item \textit{optimal order infinite memory smallest communication task respects memory restriction} ($OOSCMR$): A task with the smallest communication time is chosen from the set.
	\item \textit{optimal order infinite memory maximum accelerated task respects memory restriction} ($OOMAMR$): A task  with the highest ratio of computation time to communication time is chosen from the set.
\end{enumerate}

\begin{table}[htb]
	\begin{center}
		
		\begin{tabular}{|c|c|c|c|}
			\hline
			\multirow{2}{*}{Task} & Memory Req & \multirow{2}{*}{Comm Time} & \multirow{2}{*}{Comp Time}\\  
			&=Comm Volume && \\ \hline
			A & 4 & 4 &  1\\ \hline
			B & 2 & 2 & 6\\ \hline
			C & 8 & 8 & 8\\ \hline
			D & 5 & 5 & 4\\ \hline
			E & 3 & 3 & 2\\ \hline
		\end{tabular}
		\caption{\label{tab:staticOrderDynamicCorrectionsExample} A task set for static order dynamic corrections schedules.}
	\end{center}
\end{table}

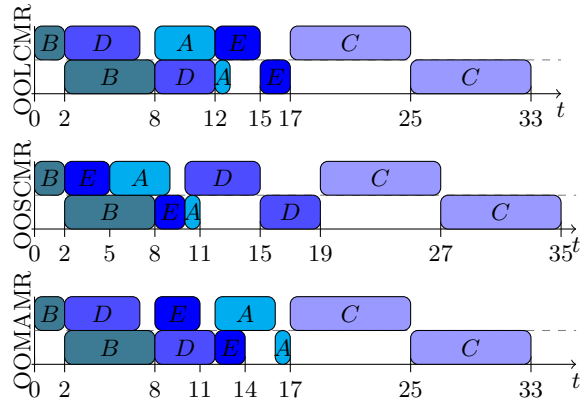
\begin{figure}[htb]
	\begin{center}
	\newcommand{\taskA}[2][0]{\task[#1]{$A$}{#2}{4}{1}{cyan}}
	\newcommand{\taskB}[2][0]{\task[#1]{$B$}{#2}{2}{6}{cyan!50!black}}
	\newcommand{\taskC}[2][0]{\task[#1]{$C$}{#2}{8}{8}{blue!40!white}}
	\newcommand{\taskD}[2][0]{\task[#1]{$D$}{#2}{5}{4}{blue!70!white}}
	\newcommand{\taskE}[2][0]{\task[#1]{$E$}{#2}{3}{2}{blue}}
	\begin{tikzpicture}[yscale=0.6,xscale=0.2]
	\begin{scope}
	\schedule{35}{OOLCMR}{2,8,12,15,17,25,33}
	\taskB{0}
	\taskD[1]{2}
	\taskA{8}
	\taskE{12}
	\taskC{17}
	\end{scope}
	\begin{scope}[yshift=-3cm]
	\schedule{36}{OOSCMR}{2,5,8,11,15,19,27,35}
	\taskB{0}
	\taskE[3]{2}
	\taskA[1]{5}
	\taskD{10}
	\taskC{19}
	\end{scope}
	\begin{scope}[yshift=-6cm]
	\schedule{36}{OOMAMR}{2,8,11,14,17,25,33}
	\taskB{0}
	\taskD[1]{2}
	\taskE[1]{8}
	\taskA{12}
	\taskC{17}
	\end{scope}
	
	\end{tikzpicture}
	\caption{ \label{fig:staticOrderDynamicCorrectionsExample} Different static order dynamic corrections heuristic schedules for a task set of Table~\ref{tab:staticOrderDynamicCorrectionsExample} with a memory capacity of 9. The $OMIM$ order is $BCDAE$.}
	\end{center}
\end{figure}


As previously, we propose on Table~\ref{tab:staticOrderDynamicCorrectionsExample} an example instance for this class of strategies, and provide on Figure~\ref{fig:staticOrderDynamicCorrectionsExample} the corresponding schedules.

\subsection{Additional heuristics from previous work}

We also consider two other static heuristics for evaluation. The first heuristic is based on an algorithm, proposed by Gilmore and Gomory, to obtain the minimal cost sequence for a set of jobs~\cite{Gilmore-Gomory:1964}. This is a classical algorithm for 2-machine no-wait flow shop. In this algorithm, each job has a start and end state and a cost is associated to change the state. In our context, this cost can be seen as non-overlap time of computation for two adjacent tasks. Here is the main idea of this algorithm. Initially, a partial sequence of jobs is represented by a graph such that their overlap is maximum. Subsequently edges are greedily added to this graph to connect two components while minimizing the total non-overlap cost. When the graph is  connected, an edge interchange mechanism is used to determine the sequence of jobs, which ensures that the sequence has minimal cost. More details can be found in the original paper~\cite{Gilmore-Gomory:1964} and our implementation is publicly available~\cite{gitworkrepo}. This algorithm does not take memory constraints into account and only provides the sequence of processing. We use this sequence with a memory capacity restriction just like for other static heuristcs, and we call this heuristic \textit{Gilmore-Gomory} ($GG$).

The second heuristic is based on the First-Fit algorithm for the bin packing problem. The idea of this heuristic is to identify groups of tasks which can fit in memory together, called \emph{bins}. In First-Fit, tasks are considered in an arbitrary order and added to the first bin in which they can fit. If no suitable bin is found then a new bin is created and this task is added to it. When all tasks have been assigned to bins, we consider the sequence made of all tasks from the first bin, then tasks for the second bin, and so on. We call this heuristic \textit{Bin Packing} ($BP$).

\subsection{Solving Linear Program Iteratively}
\label{subsec:linearprogrammingformulation}
We use a mixed integer linear program to obtain the order of communications and computations. Recall that $CP_i$ and $CM_i$ represent computation and communication times of task $i$, and the memory capacity of the target system is $C$. In the linear program formulation, $s_i$ and $e_i$ (resp. $s'_i$ and $e'_i$) represent the start and end times of communication (resp. computation) for task $i$, and $MC(i)$ is the memory capacity requirement of task $i$. The formulation also contains for each pair of tasks $i$ and $j$ \textit{i)} a boolean variable $a_{ij}$ to denote the order of $i$ and $j$ on the communication resource \textit{ii)} a boolean variables $b_{ij}$ to denote the order of $i$ and $j$ on the computation resource, and \textit{iii)} a boolean variables $c_{ij}$ to denote the order of $s_i$ and $e'_j$.



\noindent Let $L=\sum_i (CP_i + CM_i)$. It is evident that $e_i =s_i + CM_i$ and $e'_i =s'_i + CP_i$. The linear program is given below.

\vspace*{-0.5cm}
\begin{align*}
	& \text{Minimize $l$ subject to:}\\
	\forall i, \quad & e'_i \leq l && \text{(task $i$ completes)}\\
	\forall i, \quad & e_i \leq s'_i&& \text{(task $i$ valid ordering)}\\
	\forall i, \forall j\ne i, \quad & \left\{\begin{aligned}
		e_j & \leq s_i +(1-a_{ij})L\\
		e_i & \leq s_j +a_{ij}L 
	\end{aligned} \right. &&
	\begin{aligned}
		&\text{(exclusive use of}\\
		&\text{  communication link)}
	\end{aligned}\\
	\forall i, \forall j\ne i, \quad & \left\{
	\begin{aligned}
		e'_j &\leq s'_i +(1-b_{ij})L\\
		e'_i &\leq s'_j +b_{ij}L
	\end{aligned}\right. && 
	\begin{aligned}
		&\text{(exclusive use of}\\
		&\text{  computation resource)}
	\end{aligned}\\
	\forall i, \forall j\ne i, \quad &\left\{
	\begin{aligned}
		e'_j &\leq s_i +(1-c_{ij})L\\
		s_i  &< e'_j +c_{ij}L\\
	\end{aligned}\right.&&
	\begin{aligned}
		&\text{(respect ordering}\\
		&\text{ of variables $c_{ij}$)}
	\end{aligned}\\
	\forall i, \quad & \sum_{r\ne i} (a_{ir} - c_{ir})MC(r) + MC(i) \le C && \text{(memory constraint)}\\
\end{align*}

We use GLPK solver v4.65 to solve the above formulation. We also add the following constraints to help the solver: $\forall i, \forall j\ne i$, $a_{ij} + a_{ji}=1$, $b_{ij} + b_{ji} = 1$, $c_{ij} \le a_{ij}$, $c_{ij} \le b_{ij}$, and  $c_{ij} + c_{ji} \le 1$. The solver was unable to  solve this MILP at the scale of our interest in limited time. Hence, we solve the linear program iteratively for a small subset of size $k=3, 4, 5, 6$: at the boundary of two iterations we fix the event (communication or computation) of an unfinished task started before the boundary point and consider other events flexible. The subsets are formed in the order in which tasks are submitted, which is arbitrary. For a given size $k$, we represent the makespan calculated by this heuristic as $lp.k$. We compute various $lp.k$ values for different memory capacities and observe that most of the other heuristics perform better than this heuristic. Figure~\ref{fig:iterativeLpSolution} shows the performance of different heuristics with MILP based heuristics for various memory capacities of a single trace file.

			\begin{figure}[htb]
				\begin{center}
				\includegraphics[scale=0.5]{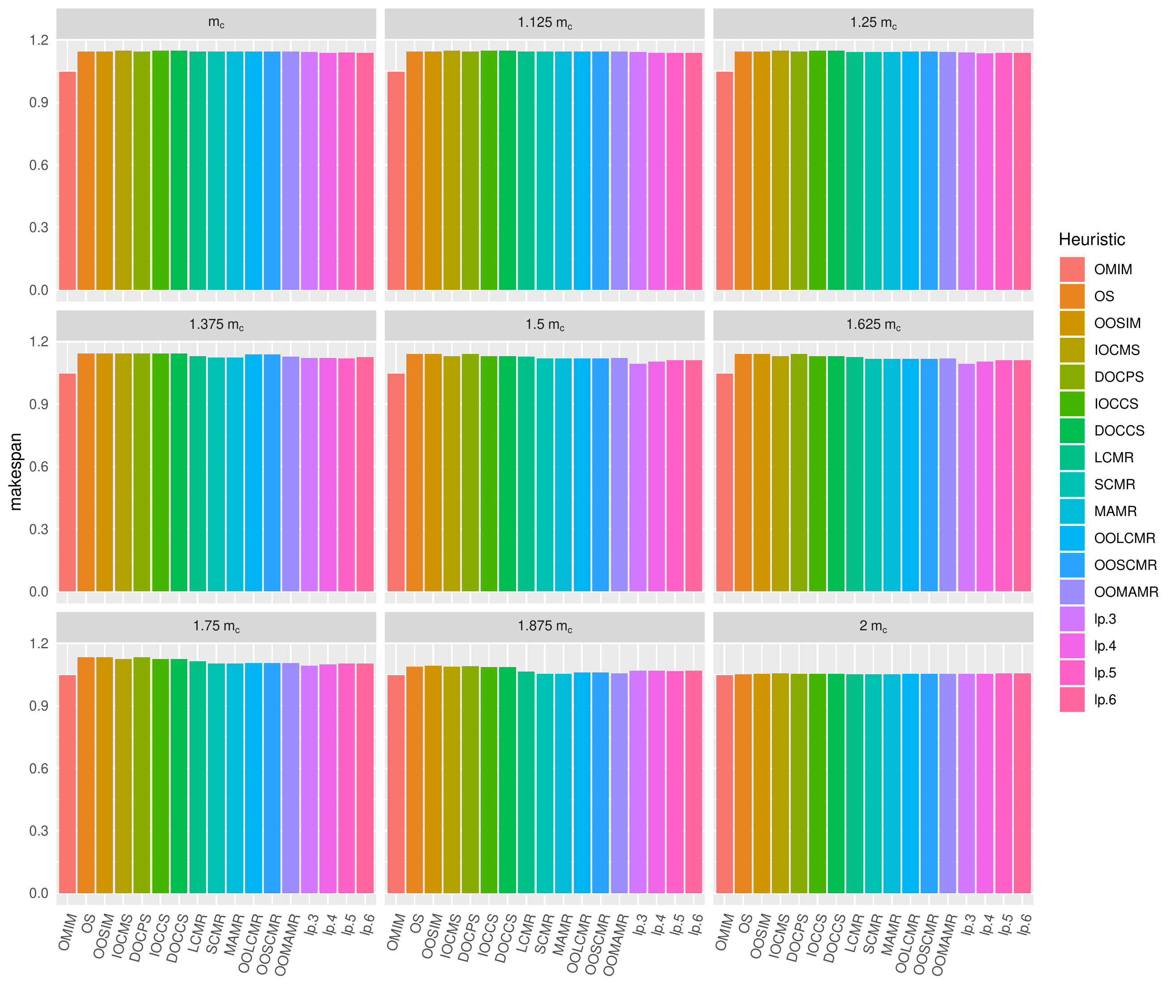}
				\caption{Comparision of proposed heuristics heuristics with MILP solution based heuristic for different memory capacities of a single trace file . Here  minimum memory requirement to process all tasks is $m_c=176KB$.}
				\label{fig:iterativeLpSolution}
				\end{center}
			\end{figure}
	
	\subsection{Favorable Situations for Heuristics}
		Based on the definition of proposed strategies and the optimality of Algorithm~\ref{alg:OrderOfExecutionInfinteMemory}, we present the scenarios which should be more favorable for each heuristic in Table~\ref{tab:heuristicsAndFavorableScenarios}. This allows programmers to use appropriate strategies to maximize communication-computation overlap for their applications. In this table, ``moderate memory capacity'' refers to the case where memory is constrained, but close to the maximal memory requirement of the $OMIM$ schedule.
		
	\begin{table}[htb]
		\begin{center}
		\begin{tabular}{|c|p{10.25cm}|}
			\hline
			\textbf{Heuristic} & \textbf{\hspace{2cm}Favorable Situation} \\ \hline
			$OOSIM$ & Memory capacity is not a restriction (\textcolor{green}{Optimal}) \\ \hline
			$IOCMS$ & Memory capacity is not a restriction and tasks are compute intensive (\textcolor{green}{Optimal}) \\ \hline
			$DOCPS$ & Memory capacity is not a restriction and tasks are communication intensive (\textcolor{green}{Optimal}) \\ \hline
			$IOCCS$ & Moderate memory capacity and most tasks are highly compute intensive \\ \hline
			$DOCCS$ & Moderate memory capacity and most tasks are highly communication intensive \\ \hline
			$LCMR$ & Limited memory capacity and significant percentage of tasks with large communication times are compute intensive\\ \hline
			$SCMR$ & Limited memory capacity and significant percentage of tasks with small communication times are compute intensive\\ \hline
			$MAMR$ & Limited memory capacity and significant percentage of tasks of both types\\ \hline
			$OOLCMR$ & Moderate memory capacity and significant percentage of tasks are communication intensive\\ \hline
			$OOSCMR$ & Moderate memory capacity and significant percentage of tasks are compute intensive \\ \hline
			$OOMAMR$ & Moderate memory capacity and significant percentage of highly compute and communication intensive tasks \\ \hline
		\end{tabular}\caption{~\label{tab:heuristicsAndFavorableScenarios}Heuristics and their favorable scenarios.}
		\end{center}
	\end{table}

			Some of these favorable scenarios can be clearly observed in our experimental results, on Figures~\ref{fig:ratio_to_optimal_hf} and~\ref{fig:ratio_to_optimal_ccsd}. For example, HF compute intensive tasks have small communication times, which explains why the $SCMR$ heuristic exhibits very good performance in limited memory cases. CCSD has significant percentage of both types of tasks, and indeed the performance of $OOLCMR$ and $OOSCMR$ is very close to optimality in moderate memory cases. 
	
	\section{Experimental Settings}
	\label{sec:expSetting}

We consider a machine called Cascade~\cite{Cascade}, available at PNNL, for our experiments. We obtain traces by running two molecular chemistry applications, double precision version of  HF and CCSD of NWChem~\cite{NWChem} package on 10 nodes of this machine. Each node is composed of 16 Intel Xeon E5-2670 cores. NWChem takes advantages of a Partitioned Global Address Space Programming Model called Global Arrays (GA)~\cite{GlobalArray} to use shared-memory programming APIs on distributed memory computers. GA dedicates one core of each node to handle other cores, hence we can view a node as being composed of 15 computational cores. We use 150 processes for each application and obtain 150 trace files. We run CCSD with Uracil molecules input and HF with SiOSi molecules (for Uracil molecules, HF has a much smaller workload, each processor executes only around 20 tasks, that is why we chose SiOSi input for HF execution). Each process executes around 300-800 tasks. Our data transfer model is quite simple and we consider that all data transfers between the local memory of each process and the GA memory take the same route. Modeling of different routes of data transfers for the same source-destination pair, bandwidth sharing for different source-destination pairs and network congestion is more challenging and part of our future work. This simple model is enough to provide insight to the application developers (or runtime system) about the ordering of data transfers for the same source-destination pair so as to maximize communication-computation overlap. Our model is easily adaptable to any source-destination pair when there is one fixed route between source and destination (such as between CPU and GPU, one copy engine to transfer data from CPU (resp. GPU) to GPU (resp. CPU) ).

Both applications mainly perform two types of computations, tensor transpose and tensor contraction. HF expects to specify a tile size and we set it to 100, so that each core can be used efficiently. CCSD automatically determines tile sizes at different program points based on the input molecules. Hence, HF operates on almost homogeneous tiles while CCSD uses more heterogeneous tiles.
	
		\begin{figure}[htb]
		\begin{center}
			\includegraphics[scale=0.5]{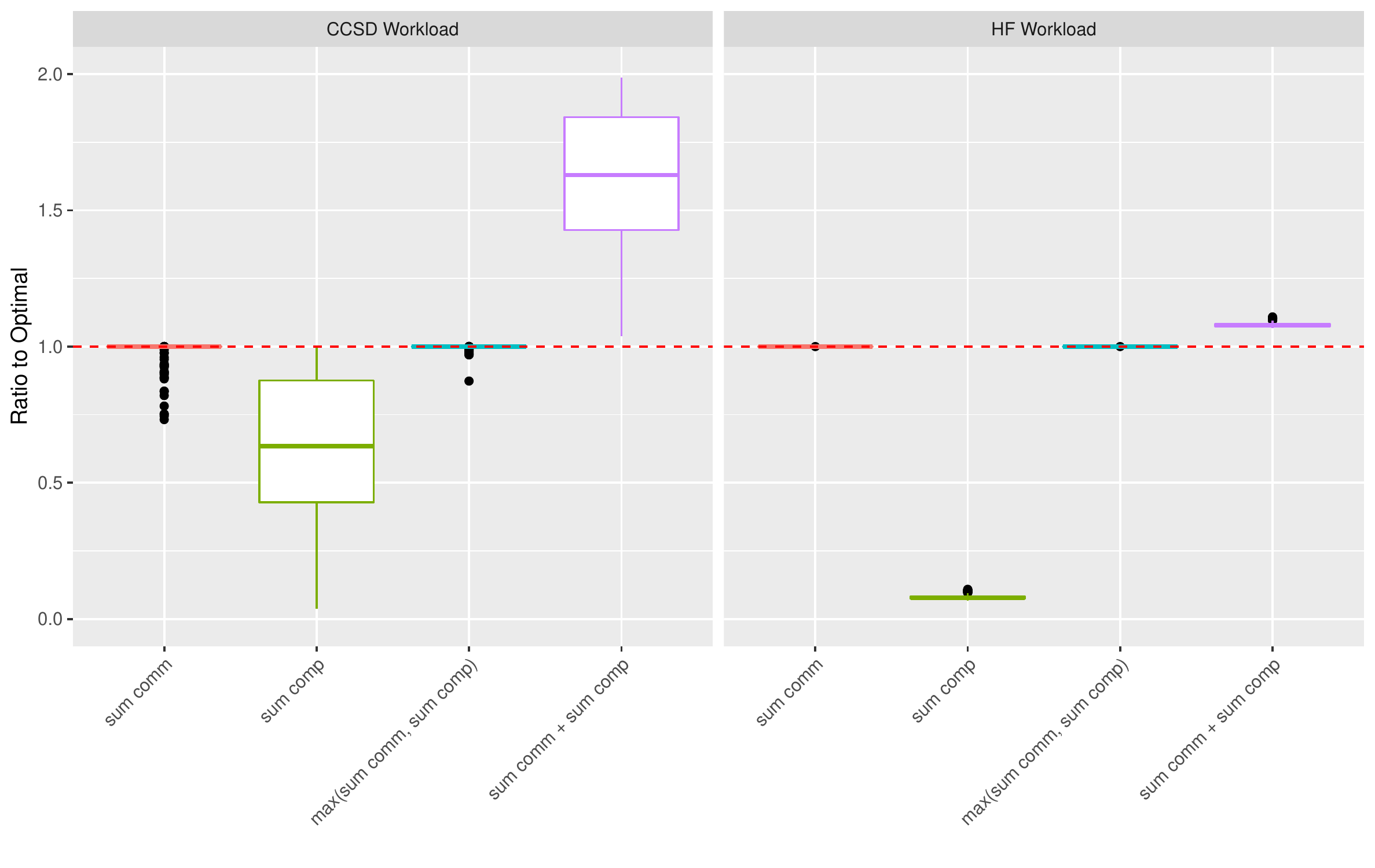}
			\caption{HF and CCSD tasks characteristics.}
			\label{fig:ApplicationProperties}
		\end{center}
		\end{figure}	
	
	\subsection{Workload Characteristics}

		To get more insights into the considered workloads, we provide  information about the instances we consider in Figure~\ref{fig:ApplicationProperties}. For each instance, we compute the sum of communication times (\textit{sum comm}) and sum of computation times  (\textit{sum comp}), and normalize it relatively to the $OMIM$ value. Figure~\ref{fig:ApplicationProperties} also shows the maximum of both values, which is a lower bound on the possible makespan of a schedule for this instance, and their sum, which is an upper bound: this represents the makespan of the sequential schedule, obtained with zero overlap between computation and communication. We can see that HF is a communication intensive application and at most 20\% overlap can be expected in the best scenario. On the other hand, in the CCSD workload, communications and computations are almost evenly distributed and a more significant overlap is possible.

	\section{Experimental Results}
	\label{sec:expResults}
	
		\begin{figure*}[htb]
		\includegraphics[scale=0.5]{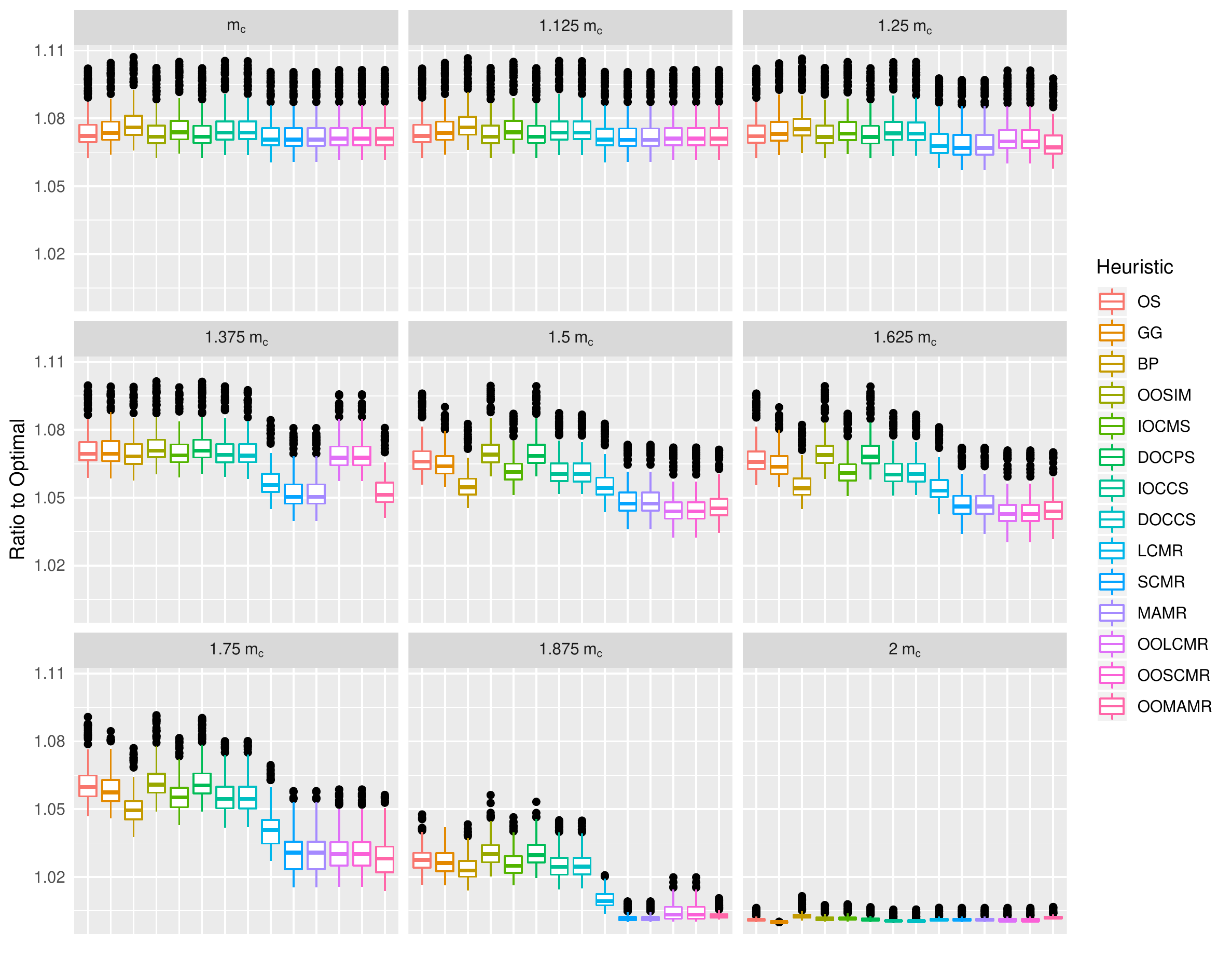}
		\caption{Comparison of different heuristics for HF with $m_c=176KB$.}
		\label{fig:ratio_to_optimal_hf}
		\end{figure*}	
		We evaluate our scheduling heuristics for several memory capacities. From the obtained traces, we first determine the minimum requirement of the memory capacity $m_c$ to execute all tasks. Then we observe the behavior of all heuristics with memory capacity $m_c$ to $2m_c$, in increments of $0.125m_c$. Our performance metric is the ratio to optimal $r$: if heuristic $H$ has makespan $M_H$ on an instance, and the optimal makespan for the infinite memory case is $OMIM$, then $r(H)=\frac{M_H}{OMIM}$ (lower values are better). This ratio is at least $1$, and a value close to $1$ indicates a well-suited  heuristic which achives maximum possible communication-computation overlap.

Figures ~\ref{fig:ratio_to_optimal_hf} and ~\ref{fig:ratio_to_optimal_ccsd} show the distribution of the performance of each heuristic for different memory capacities, where plots are categorized by memory capacities. For each memory capacity and each heuristic, the box on the plot displays the median, first and last quartile, and the whiskers indicate minimum and maximum values, with outliers are shown by black dots.
	\subsection{HF Performance}	
	
	\begin{figure}[htb]
	\begin{center}
		\includegraphics[scale=0.7]{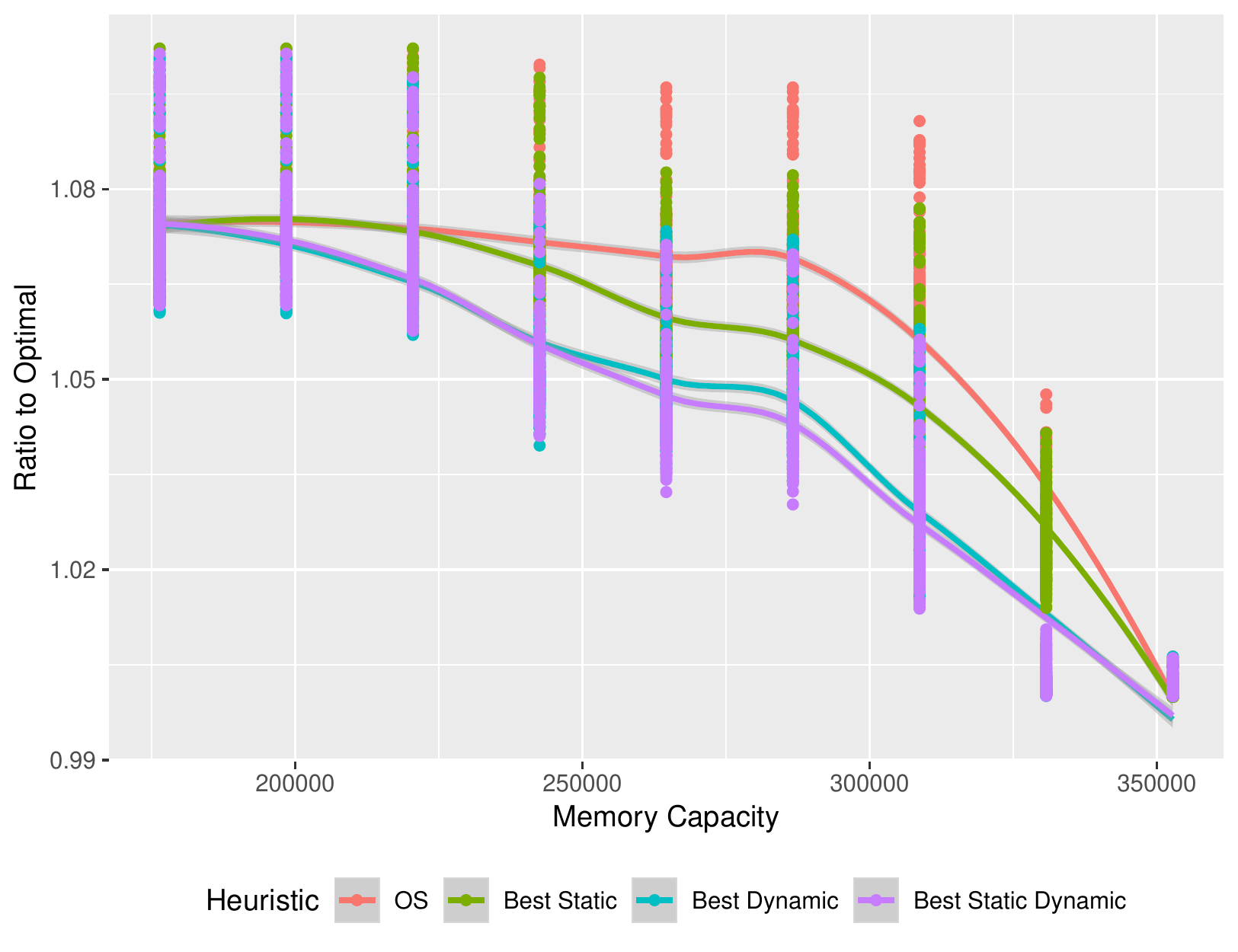}
		\caption{Comparision of best variants of all categories for HF.}
		\label{fig:ratio_to_optimal_best_hf}
	\end{center}
	\end{figure}
	
		As indicated above, HF tasks operate on less heterogeneous tiles, this is also noticeable in Figure~\ref{fig:ratio_to_optimal_hf}. All heuristics depict similar behavior for minimum memory capacity $m_c$ and increasing the memory capacity slightly does not change the performance of all heuristics. As memory capacity is increased further, dynamic variants of heuristics start performing better. For the moderate memory capacities (close to $2m_c$), static order with dynamic corrections variants outperform others. $GG$ heuristic does not achieve good performance, because its task sequence is obtained considering no extra memory is available, but is then applied in a different scenario where memory is limited. Surprisingly, the $BP$ heuristic which considers only memory constraint obtains good performance for a static heuristic, but is outperformed by more dynamic approaches.

Figure~\ref{fig:ratio_to_optimal_best_hf} shows the performance comparison of the best variant in each category, in addition to the \textit{order of submission} ($OS$) strategy which orders tasks in the (arbitrary) sequence in which they are given. Static strategies are expected to perform better when there is not any memory capacity restriction, and indeed this plot shows that static strategies face capacity bottleneck and underperform with limited memory. Dynamic strategies achieve slightly better performance with limited memory capacity, but when memory capacity is larger, static order with dynamic corrections strategies outperform all others.
		\begin{figure}[htb]
		\begin{center}
			\includegraphics[scale=0.5]{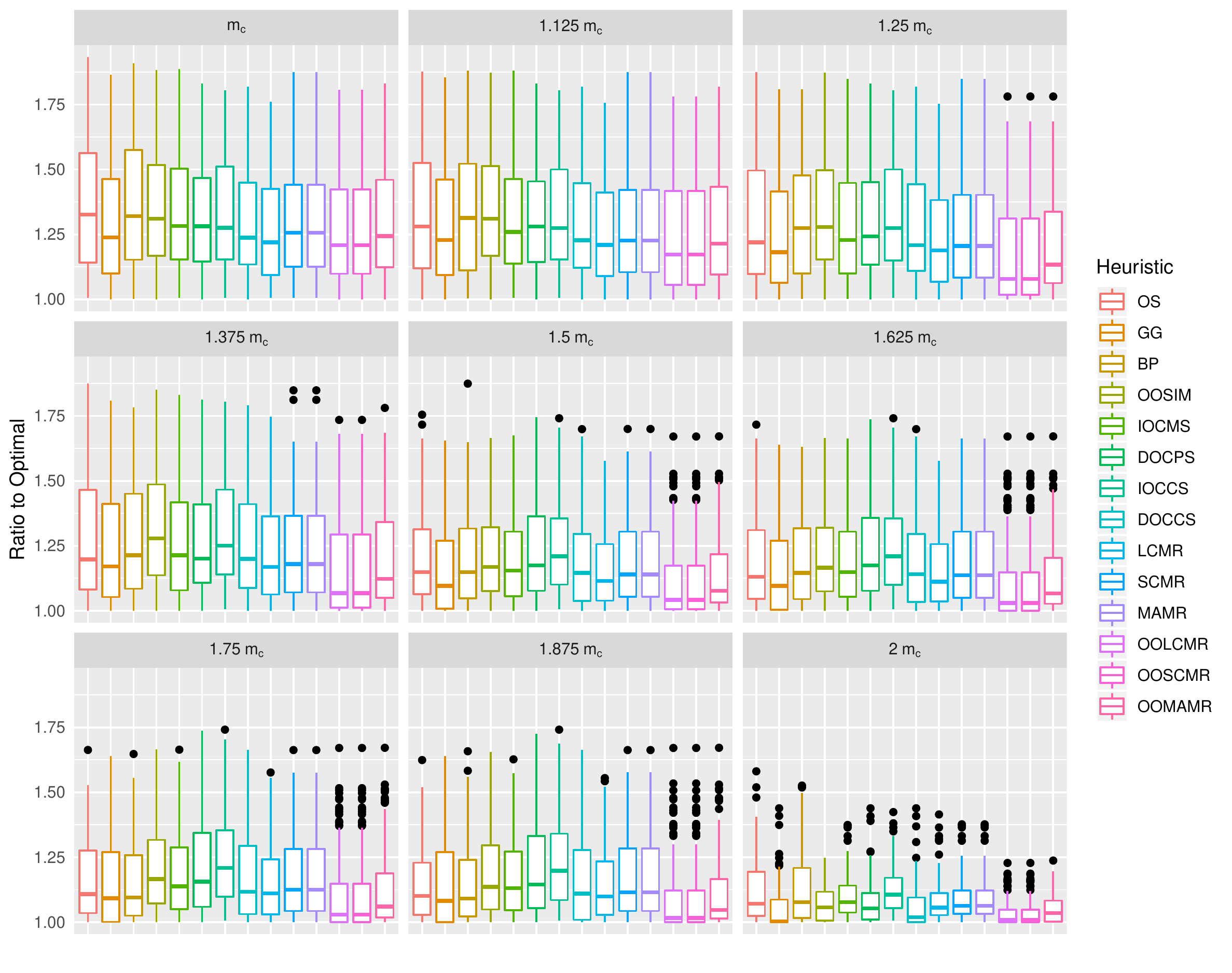}
			\caption{Comparison of different heuristics for CCSD with $m_c=1.8GB$.}
			\label{fig:ratio_to_optimal_ccsd}
	\end{center}
	\end{figure}	
	
	\subsection{CCSD Performance}

		The CCSD application operates on tasks of different sizes, hence different heuristics exhibit distinct behaviors even at minimum memory capacity $m_c$. Heterogeneity favors dynamic strategies, as can be seen by the fact that both dynamic and static order with dynamic corrections based strategies perform better than static based strategies. Similar to HF, static order with dynamic corrections based strategies outperform others as memory capacity becomes moderate.

	\begin{figure}[htb]
		\begin{center}
		\includegraphics[scale=0.7]{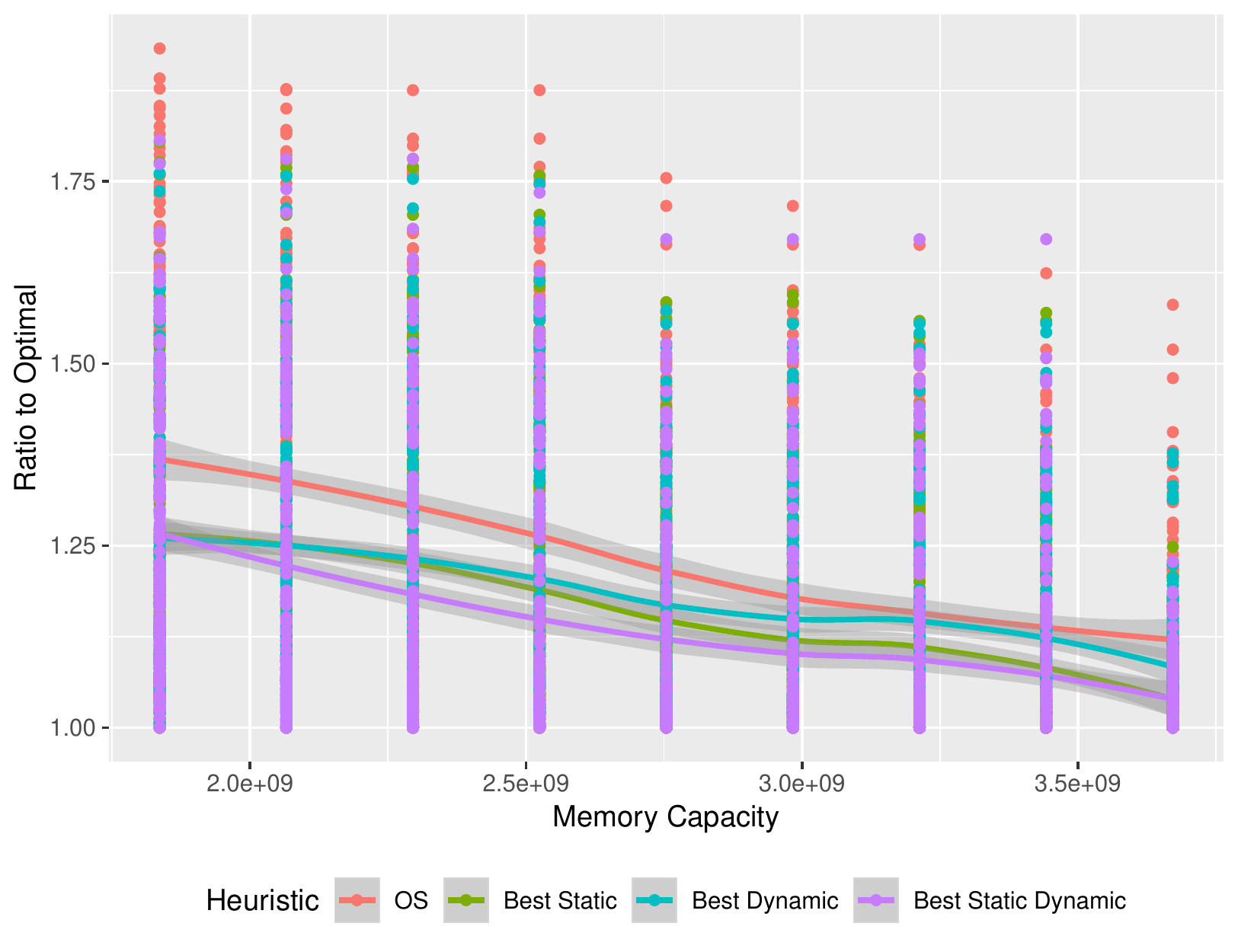}
		\caption{Comparision of best variants of all categories for CCSD.}
		\label{fig:ratio_to_optimal_best_ccsd}
		\end{center}
	\end{figure}
	
			\begin{figure}[!ht]
		\begin{center}
			\includegraphics[scale=0.65]{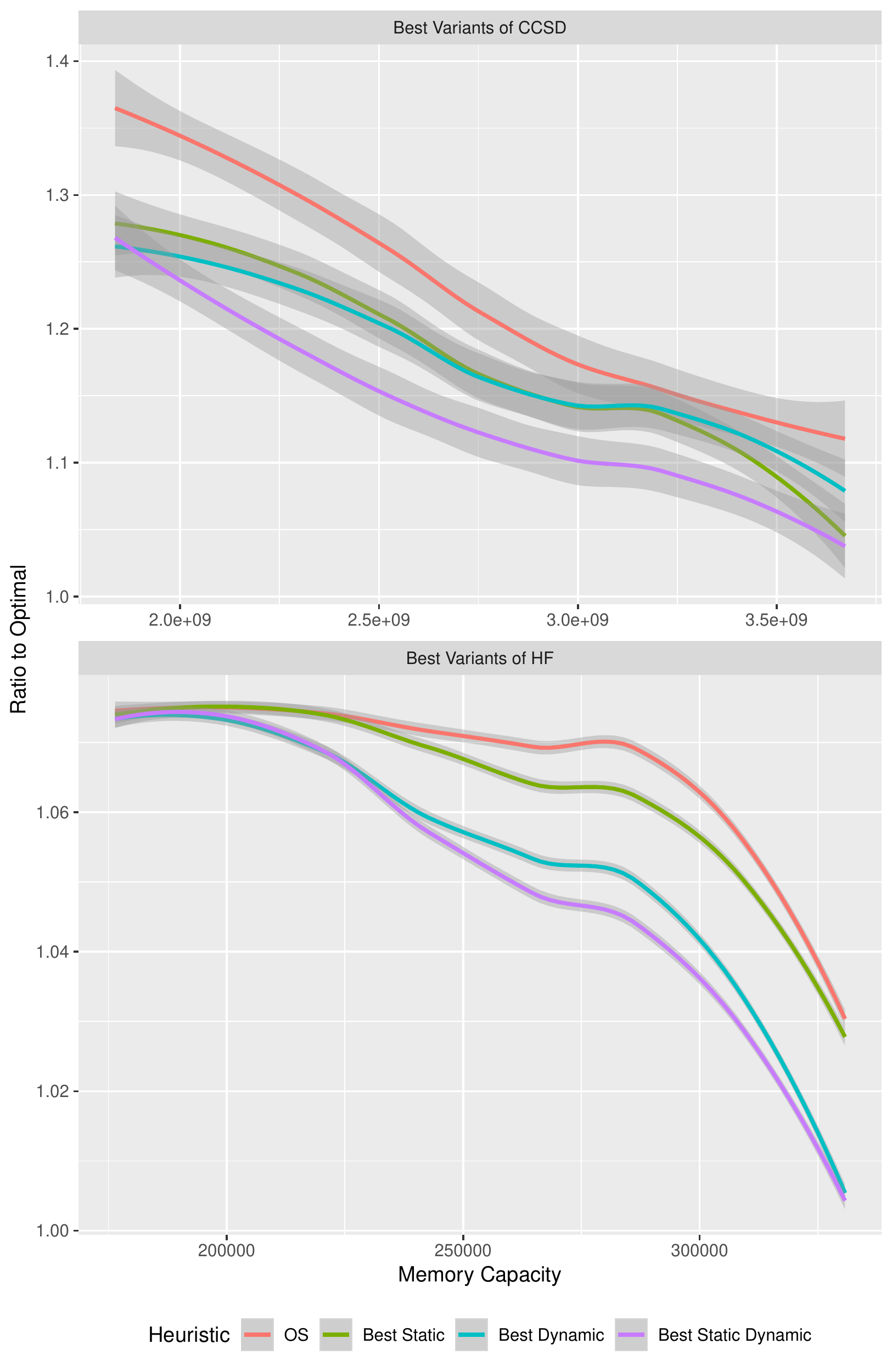}
			\caption{Best variants of all categories where heuristics are applied in the batches of 100 tasks.}
			\label{fig:best_variants_batch}
		\end{center}
		\end{figure}
	Figure~\ref{fig:ratio_to_optimal_best_ccsd} shows that best variants of dynamic and static order with dynamic corrections strategies achieve similar performance at minimum memory capacity $m_c$. But as memory capacity increases, heterogeneity allows static order with dynamic corrections based strategies to take advantage of static knowledge to get maximum overlap and dynamic correction to select another task in case of memory capacity limitation. Static strategies also start performing better at the end, which indicates that this application has potential for significant communication-computation overlap and pure dynamic strategies are unable to take this information into account while making scheduling decisions.

	\subsection{Scheduling in Batches}
		In most applications, the runtime scheduler may only observe a limited batch of independent tasks. Therefore we organize tasks of each trace file in batches of 100 (the last batch may have less than 100 tasks). We apply each heuristic on each group in succession. Figure~\ref{fig:best_variants_batch} shows the performance of the best variants of each category for both applications. The plots exhibit behavior similar to Figures~\ref{fig:ratio_to_optimal_best_hf} and~\ref{fig:ratio_to_optimal_best_ccsd}: static order with dynamic corrections variants attain maximum communication-computation overlap and outperform other heuristics.
	

	\section{Conclusion and Perspectives}
	\label{sec:conclusion}
	
	
		In this article, we consider the problem of deciding the order of data transfers between two memory nodes such that overlap of communications and computations is maximized. With Exascale computing, applications face bottlenecks due to communications. Hence, it is extremely important to achieve the maximum  communication-computation overlap in order to exploit the full potential of the system. We show that determining the order of data transfers is a NP complete problem. We propose several data transfer heuristics and evaluate them on two molecular chemistry kernels, HF and CCSD. Our results show that some of our heuristics achieve significant overlap and perform very close to the lower bound. We plan to evaluate our strategies on different applications coming from multiple domains. We also plan to study the behavior of our strategies in the context of overlapping CPU-GPU communications with computations. A runtime system aiming at exposing different heuristics to maximize the communication-computation overlap at the developer level and automatically selecting the best one is currently underway.

\bibliographystyle{IEEEtran}
\bibliography{article}

\begin{thebibliography}{10}
\providecommand{\url}[1]{#1}
\csname url@samestyle\endcsname
\providecommand{\newblock}{\relax}
\providecommand{\bibinfo}[2]{#2}
\providecommand{\BIBentrySTDinterwordspacing}{\spaceskip=0pt\relax}
\providecommand{\BIBentryALTinterwordstretchfactor}{4}
\providecommand{\BIBentryALTinterwordspacing}{\spaceskip=\fontdimen2\font plus
\BIBentryALTinterwordstretchfactor\fontdimen3\font minus
  \fontdimen4\font\relax}
\providecommand{\BIBforeignlanguage}[2]{{%
\expandafter\ifx\csname l@#1\endcsname\relax
\typeout{** WARNING: IEEEtran.bst: No hyphenation pattern has been}%
\typeout{** loaded for the language `#1'. Using the pattern for}%
\typeout{** the default language instead.}%
\else
\language=\csname l@#1\endcsname
\fi
#2}}
\providecommand{\BIBdecl}{\relax}
\BIBdecl

\bibitem{webpagescheduling}
P.~Brucker and S.~Knust, ``Complexity results for scheduling problems,'' Web
  document, URL: \url{http://www2.informatik.uni-osnabrueck.de/knust/class/}.

\bibitem{bleuse2015scheduling}
R.~Bleuse, S.~Kedad-Sidhoum, F.~Monna, G.~Mouni{\'e}, and D.~Trystram,
  ``Scheduling independent tasks on multi-cores with gpu accelerators,''
  \emph{Concurrency and Computation: Practice and Experience}, vol.~27, no.~6,
  pp. 1625--1638, 2015.

\bibitem{YarKhan:2011:Quark:Manual}
A.~YarKhan, J.~Kurzak, and J.~Dongarra, \emph{QUARK Users' Guide: QUeueing And
  Runtime for Kernels}, UTK ICL, 2011.

\bibitem{parsec}
G.~Bosilca, A.~Bouteiller, A.~Danalis, M.~Faverge, T.~H{\'e}rault, and
  J.~Dongarra, ``{PaRSEC: A programming paradigm exploiting heterogeneity for
  enhancing scalability},'' \emph{{Computing in Science and Engineering}},
  2013.

\bibitem{starpu}
\BIBentryALTinterwordspacing
C.~Augonnet, S.~Thibault, R.~Namyst, and P.-A. Wacrenier, ``{StarPU: A Unified
  Platform for Task Scheduling on Heterogeneous Multicore Architectures},''
  \emph{Concurrency and Computation: Practice and Experience, Special Issue:
  Euro-Par 2009}, vol.~23, pp. 187--198, Feb. 2011. [Online]. Available:
  \url{http://hal.inria.fr/inria-00550877}
\BIBentrySTDinterwordspacing

\bibitem{legion12}
\BIBentryALTinterwordspacing
M.~Bauer, S.~Treichler, E.~Slaughter, and A.~Aiken, ``Legion: Expressing
  locality and independence with logical regions,'' in \emph{Proceedings of the
  International Conference on High Performance Computing, Networking, Storage
  and Analysis}, ser. SC '12.\hskip 1em plus 0.5em minus 0.4em\relax Los
  Alamitos, CA, USA: IEEE Computer Society Press, 2012, pp. 66:1--66:11.
  [Online]. Available: \url{http://dl.acm.org/citation.cfm?id=2388996.2389086}
\BIBentrySTDinterwordspacing

\bibitem{ompss}
\BIBentryALTinterwordspacing
A.~Duran, E.~Ayguad{\'{e}}, R.~M. Badia, J.~Labarta, L.~Martinell,
  X.~Martorell, and J.~Planas, ``Ompss: a proposal for programming
  heterogeneous multi-core architectures,'' \emph{Parallel Processing Letters},
  vol.~21, no.~2, pp. 173--193, 2011. [Online]. Available:
  \url{https://doi.org/10.1142/S0129626411000151}
\BIBentrySTDinterwordspacing

\bibitem{kaapi}
E.~Hermann, B.~Raffin, F.~Faure, T.~Gautier, and J.~Allard, ``{Multi-{GPU} and
  Multi-{CPU} Parallelization for Interactive Physics Simulations},'' in
  \emph{Euro-Par (2)}, 2010, pp. 235--246.

\bibitem{ascaccommitteereport2014}
``Top ten exascale research challenges,'' ASCAC committee report, URL:
  \url{https://science.energy.gov/\~/media/ascr/ascac/pdf/meetings/20140210/Top10reportFEB14.pdf},
  2014.

\bibitem{yelick2016}
\BIBentryALTinterwordspacing
K.~Yelick, ``Avoiding, hiding and managing communication at the exasacle,''
  2016. [Online]. Available:
  \url{https://people.eecs.berkeley.edu/~yelick/talks/exascale/Communication-Yelick-China16.pdf}
\BIBentrySTDinterwordspacing

\bibitem{Papadimitriou:1980:FSL:322203.322213}
\BIBentryALTinterwordspacing
C.~H. Papadimitriou and P.~C. Kanellakis, ``Flowshop scheduling with limited
  temporary storage,'' \emph{J. ACM}, vol.~27, no.~3, pp. 533--549, Jul. 1980.
  [Online]. Available: \url{http://doi.acm.org/10.1145/322203.322213}
\BIBentrySTDinterwordspacing

\bibitem{GareyJohnson}
M.~R. Garey and D.~S. Johnson, \emph{Computers and Intractability, a Guide to
  the Theory of {NP}-Completeness}.\hskip 1em plus 0.5em minus 0.4em\relax W.H.
  Freeman and Company, 1979.

\bibitem{heft-Topcuoglu}
\BIBentryALTinterwordspacing
H.~Topcuouglu, S.~Hariri, and M.-y. Wu, ``{Performance-Effective and
  Low-Complexity Task Scheduling for Heterogeneous Computing},'' \emph{IEEE
  Trans. Parallel Distrib. Syst.}, vol.~13, no.~3, pp. 260--274, Mar. 2002.
  [Online]. Available: \url{http://dx.doi.org/10.1109/71.993206}
\BIBentrySTDinterwordspacing

\bibitem{hipc16multiresource}
\BIBentryALTinterwordspacing
O.~Beaumont, T.~Cojean, L.~Eyraud-Dubois, A.~Guermouche, and S.~Kumar,
  ``{Scheduling of Linear Algebra Kernels on Multiple Heterogeneous
  Resources},'' in \emph{{International Conference on High Performance
  Computing, Data, and Analytics (HiPC 2016)}}, Hyderabad, India, Dec. 2016.
  [Online]. Available: \url{https://hal.inria.fr/hal-01361992}
\BIBentrySTDinterwordspacing

\bibitem{ipdps16starpu}
\BIBentryALTinterwordspacing
E.~Agullo, O.~Beaumont, L.~Eyraud{-}Dubois, and S.~Kumar, ``{Are Static
  Schedules so Bad? {A} Case Study on Cholesky Factorization},'' in \emph{2016
  {IEEE} International Parallel and Distributed Processing Symposium, {IPDPS}
  2016, Chicago, IL, USA, May 23-27, 2016}, 2016, pp. 1021--1030. [Online].
  Available: \url{http://dx.doi.org/10.1109/IPDPS.2016.90}
\BIBentrySTDinterwordspacing

\bibitem{luka-dmdar}
L.~Stanisic, S.~Thibault, A.~Legrand, B.~Videau, and J.-F. M{\'e}haut,
  ``Modeling and simulation of a dynamic task-based runtime system for
  heterogeneous multi-core architectures,'' in \emph{Euro-Par 2014 Parallel
  Processing}, F.~Silva, I.~Dutra, and V.~Santos~Costa, Eds.\hskip 1em plus
  0.5em minus 0.4em\relax Cham: Springer International Publishing, 2014, pp.
  50--62.

\bibitem{agullo_fmm}
\BIBentryALTinterwordspacing
E.~Agullo, B.~Bramas, O.~Coulaud, E.~Darve, M.~Messner, and T.~Takahashi,
  ``Task-based fmm for heterogeneous architectures,'' \emph{Concurrency and
  Computation: Practice and Experience}, vol.~28, no.~9, pp. 2608--2629, 2016.
  [Online]. Available:
  \url{https://onlinelibrary.wiley.com/doi/abs/10.1002/cpe.3723}
\BIBentrySTDinterwordspacing

\bibitem{predari:tel-01518956}
\BIBentryALTinterwordspacing
M.~Predari, ``{Load Balancing for Parallel Coupled Simulations},'' Theses,
  {Universit{\'e} de Bordeaux, LaBRI ; Inria Bordeaux Sud-Ouest}, Dec. 2016.
  [Online]. Available: \url{https://hal.inria.fr/tel-01518956}
\BIBentrySTDinterwordspacing

\bibitem{johnson}
\BIBentryALTinterwordspacing
S.~M. Johnson, ``Optimal two- and three-stage production schedules with setup
  times included,'' \emph{Naval Research Logistics Quarterly}, vol.~1, no.~1,
  pp. 61--68, 1954. [Online]. Available:
  \url{https://onlinelibrary.wiley.com/doi/abs/10.1002/nav.3800010110}
\BIBentrySTDinterwordspacing

\bibitem{Sethi:1970:GOC:321607.321620}
\BIBentryALTinterwordspacing
R.~Sethi and J.~D. Ullman, ``The generation of optimal code for arithmetic
  expressions,'' \emph{J. ACM}, vol.~17, no.~4, pp. 715--728, Oct. 1970.
  [Online]. Available: \url{http://doi.acm.org/10.1145/321607.321620}
\BIBentrySTDinterwordspacing

\bibitem{vsarkar-pact}
D.~Sbîrlea, Z.~Budimlić, and V.~Sarkar, ``Bounded memory scheduling of
  dynamic task graphs,'' in \emph{2014 23rd International Conference on
  Parallel Architecture and Compilation Techniques (PACT)}, Aug 2014, pp.
  343--355.

\bibitem{loris-ipdps18}
L.~Marchal, H.~Nagy, B.~Simon, and F.~Vivien, ``Parallel scheduling of dags
  under memory constraints,'' in \emph{2018 IEEE International Parallel and
  Distributed Processing Symposium (IPDPS)}, May 2018, pp. 204--213.

\bibitem{3machineFlowShopNPComplete}
\BIBentryALTinterwordspacing
M.~R. Garey, D.~S. Johnson, and R.~Sethi, ``The complexity of flowshop and
  jobshop scheduling,'' \emph{Math. Oper. Res.}, vol.~1, no.~2, pp. 117--129,
  May 1976. [Online]. Available: \url{http://dx.doi.org/10.1287/moor.1.2.117}
\BIBentrySTDinterwordspacing

\bibitem{Gilmore-Gomory:1964}
\BIBentryALTinterwordspacing
P.~C. Gilmore and R.~E. Gomory, ``Sequencing a one state-variable machine: A
  solvable case of the traveling salesman problem,'' \emph{Operations
  Research}, vol.~12, no.~5, pp. 655--679, 1964. [Online]. Available:
  \url{https://doi.org/10.1287/opre.12.5.655}
\BIBentrySTDinterwordspacing

\bibitem{gitworkrepo}
``Commincation scheduling,''
  \url{https://github.com/surakuma/communication-scheuling}, 2019.

\bibitem{Cascade}
\BIBentryALTinterwordspacing
``Computing: Cascade,'' 2019. [Online]. Available:
  \url{https://www.emsl.pnl.gov/emslweb/10.25582/inst.34218}
\BIBentrySTDinterwordspacing

\bibitem{NWChem}
\BIBentryALTinterwordspacing
``Nwchem: A comprehensive and scalable open-source solution for large scale
  molecular simulations,'' \emph{Computer Physics Communications}, vol. 181,
  no.~9, pp. 1477 -- 1489, 2010. [Online]. Available:
  \url{http://www.sciencedirect.com/science/article/pii/S0010465510001438}
\BIBentrySTDinterwordspacing

\bibitem{GlobalArray}
\BIBentryALTinterwordspacing
J.~Nieplocha, R.~J. Harrison, and R.~J. Littlefield, ``Global arrays: A
  nonuniform memory access programming model for high-performance computers,''
  \emph{The Journal of Supercomputing}, vol.~10, no.~2, pp. 169--189, Jun 1996.
  [Online]. Available: \url{https://doi.org/10.1007/BF00130708}
\BIBentrySTDinterwordspacing

\end{thebibliography}

\end{document}